%% file: channel_junta.tex
\documentclass[11pt]{article}

\usepackage{geometry}
\usepackage{times}
\usepackage{sty/zongbo}
\usepackage{sty/tikzit}
\usepackage{physics}
\input{tikz/sample_style.tikzstyles}
\usepackage{algpseudocode}
\usepackage{amsmath}
\usepackage{amssymb}
\usepackage{amsthm}
\usepackage[T1]{fontenc}
\usepackage[ruled]{algorithm2e}
\usepackage{thmtools}
\usepackage{thm-restate}
\definecolor{TSUYUKUSA}{RGB}{46, 169, 255}
\definecolor{KURENAI}{RGB}{203, 27, 69}
\usepackage[pdfstartview=FitH,pdfpagemode=UseNone,colorlinks=true,citecolor=KURENAI,linkcolor=TSUYUKUSA,backref=page,linktocpage=true]{hyperref}
\AtBeginDocument{%
  \DeclareFontShape{T1}{ptm}{m}{scit}{<->ssub*ptm/m/sc}{}%
}

\def\dist{\textnormal{dist}}
\def\tilde{\widetilde}

\newtheorem{theorem}{Theorem}
\newtheorem{lemma}[theorem]{Lemma}
\newtheorem{proposition}[theorem]{Proposition}

\newtheorem{corollary}[theorem]{Corollary}
\newtheorem{definition}[theorem]{Definition}

\newtheorem{fact}[theorem]{Fact}

\def\Sgm{\Sigma}
\newcommand{\vect}[1]{\text{vec}\p{#1}}

\newcommand{\dotnorm}[1]{\norm{#1}_{\diamond,1\rarr 2}}
\newcommand{\otnorm}[1]{\norm{#1}_{1\rarr 2}}
\def\Pz{\Psi_{0}}

\def\ip{\lra}
\def\wPhi{\ensuremath{\widehat{\Phi}}}
\def\wPsi{\ensuremath{\widehat{\Psi}}}
\def\inf{\mathrm{Inf}}

\def\Zfn{\ensuremath{\ZZ_4^n}}
\def\D{\mathrm{D}}
\def\0{\mathbf{0}}

\DeclareMathOperator*{\argmax}{arg\,max}
\DeclareMathOperator*{\argmin}{arg\,min}

\title{On Testing and Learning Quantum Junta Channels}
\author{
    Zongbo Bao\thanks{State Key Laboratory for Novel Software Technology, Nanjing University. Email: baozb0407@gmail.com.}
    \and Penghui Yao\thanks{State Key Laboratory for Novel Software Technology, Nanjing University. Email: phyao1985@gmail.com.}~\thanks{Hefei National Laboratory, Hefei 230088, China.}
}
\date{\today}

\begin{document}

\maketitle

\begin{abstract}
We consider the problems of testing and learning quantum  $k$-junta channels, 
    which are $n$-qubit to $n$-qubit quantum channels acting non-trivially on 
    at most $k$ out of $n$ qubits and leaving the rest of qubits unchanged. 
We show the following.
\begin{enumerate}
    \item An $O(k)$-query algorithm to distinguish whether the given channel is $k$-junta channel or is \emph{far} from any $k$-junta channels, and a lower bound $\Omega(\sqrt{k})$ on the number of queries;
    \item An $\tilde{O}\p{4^k}$-query algorithm to learn a $k$-junta channel, and a lower bound $\Omega\p{4^k/k}$ on the number of queries.
\end{enumerate}
This gives the first junta channel testing and learning results, 
    and partially answers an open problem raised by~\cite{CNY22_juntatesting}.
In order to settle these problems, we develop a Fourier analysis framework over 
    the space of superoperators and prove several fundamental properties, 
    which extends the Fourier analysis over the space of operators introduced 
    in~\cite{MO10_QBF}.

Besides, we introduce \algname{Influence-Sample} to replace \algname{Fourier-Sample}
    proposed in \cite{AS07_qtesting}. 
Our \algname{Influence-Sample} includes only single-qubit operations and 
    results in only constant-factor decrease in efficiency.
\end{abstract}

\section{Introduction}

It is crucial in quantum computing to understand the behavior of a quantum process, which is also modeled as a quantum channel, in a black-box manner. The most general method for doing this is quantum process tomography (QPT). But it requires a large amount of  computational resources, which is exponential in the number of qubits it acts on, as noted by~\cite{CN97_QPT} and~\cite{GJ14_PT}.

A quantum channel is referred to as a $k$-junta channel if it acts non-trivially on up to $k$ out of $n$ qubits, leaving the rest qubits unchanged. Characterizing a $k$-junta channel is easier if $k$ is small, hence it is interesting to find efficient algorithms to test whether a quantum channel is a $k$-junta channel and learn $k$-junta channels. The problems of testing and learning $k$-junta boolean functions is also an important problem in theoretical computer science, having a rich history of research, see~\cite{Gol17_introtopropertytesting} and~\cite{BY22_propertytesting}. More recently, testing and learning $k$-junta unitaries has been explored by~\cite{MO10_QBF,Wan11_qtestingq,CNY22_juntatesting}

In this paper, we are concerned about the testing and learning $k$-junta channels. The setting is as follows. Given oracle access to a quantum channel $\Phi$, the algorithm is supposed to output an answer about the channel $\Phi$, where \emph{access} means the algorithm queries the oracle with any $n$-qubit quantum state $\rho$ and obtains $\Phi(\rho)$ as an output. For both problems, it requires a distance function $\dist(\cdot,\cdot)$ to formulate \emph{far} and \emph{close} rigorously. With the assistance of the oracle, we are supposed to determine whether $\Phi$ is a $k$-junta channel or \emph{far} from any $k$-junta channels in the testing problem or output a description of $\tilde{\Phi}$, which is close to $\Phi$ with respect to the distance. In this work we choose the distance function induced by the inner product over superoperators, which will be formally defined in Section~\ref{sec:fourier}.

The first main result is an algorithm testing whether a given black-box channel is a $k$-junta channel or far from any $k$-junta channel with $O(k)$ queries.

\begin{theorem}[Testing Quantum $k$-Junta Channels]\label{thm:main1}
	There exists an algorithm such that, 
        given oracle access to an $n$-qubit to $n$-qubit quantum channel $\Phi$, 
        it makes $O(k/\vep^2)$ queries and determines whether $\Phi$ is a 
        $k$-junta channel or $\dist(\Phi,\Psi)\ge \vep$ for any $k$-junta 
        channel $\Psi$ with probability at least $9/10$. 
    Furthermore, Any quantum algorithm achieving this task requires 
        $\Omega(\sqrt{k})$ queries.
\end{theorem}

Our second main result is a learning algorithm, which is given a black-box $k$-junta channel and outputs a description of a channel close to the channel with $O(4^k/\vep^2)$ queries. We also show that the algorithm is almost optimal. Hence we can learn a $k$-junta channel efficiently, especially without dependence on the total number of qubits.

\begin{theorem}[Learning Quantum $k$-Junta Channels]\label{thm:main2}
There exists an algorithm, given oracle access to an $n$-qubit to $n$-qubit $k$-junta channel $\Phi$, it makes $O(4^k/\vep^2)$ queries and outputs a description of channel $\Psi$  satisfying $\dist(\Phi,\Psi)\le \vep$ with probability at least $9/10$.  Furthermore, any quantum algorithm achieving this task requires $\Omega\p{4^k/k}$ queries.
\end{theorem}

Both algorithms are proved via Fourier-analytic techniques over superoperators defined in Section~\ref{sec:fourier}. In particular, we turn both problems to estimating the influence of a superoperator, which is a generalization of the influence of boolean functions~\cite{O14_ABF} and the influence of operators~\cite{MO10_QBF}. We prove a series of fundamental properties of the Fourier analysis and the influence of superoperators extending similar results on operators. The lower bound on testing $k$-junta channels combines the result of testing boolean $k$-juntas obtained by~\cite{BKT20_lowerbound} and a structural result for $k$-junta channels. The lower bound on learning $k$-junta channels is obtained by a reduction from learning $k$-junta unitaries.

Besides, we exhibit a simple Influence-Estimator to estimate the influence of channels in Appendix~\ref{app:ie}. Compared with the estimator in~\cite{CNY22_juntatesting}, where it requires entanglement and 2-qubit operations, our Influence-Estimator requires only single-qubit operations and is as efficient as theirs. Therefore, it might be easily implemented in the lab.

\paragraph{Contributions}

\begin{enumerate}
	\item We develop Fourier analysis over superoperators and prove several basic properties around influence, which are an extension of Fourier analysis over operators~\cite{MO10_QBF} and may be of independent interest;
	\item We present the first $k$-junta channel testing algorithm and a lower bound for this problem, partially answering an open problem raised by~\cite{CNY22_juntatesting}. In addition, we show an almost optimal algorithm for $k$-junta channel learning problems;
	\item We construct a new and simple Influence-Estimator, which may be easy to implement in the lab since it includes only single-qubit operations.
\end{enumerate}

\paragraph{Organization} In Section~\ref{sec:rw}, we present a brief overview of related works. Section~\ref{sec:tech} provides a high-level overview of the proof techniques. After establishing some preliminaries regarding quantum channels in Section~\ref{sec:pre}, we demonstrate our Fourier-analytic techniques in Section~\ref{sec:fourier}, including properties of our distance function. In Section~\ref{sec:testing} and Section~\ref{sec:learning}, we prove our $k$-junta channel testing and learning results respectively. Finally, we conclude in Section~\ref{sec:summary}.

\subsection{Related Work}\label{sec:rw}

A boolean function $f:\{0,1\}^n\rightarrow \{0,1\}$ is a $k$-junta if its value only depends on at most $k$ coordinates of the inputs. Testing and learning boolean juntas has been extensively studied for decades. The first result explicitly related to testing juntas is obtained by~\cite{PRS02_testing}, where an $O(1)$-queries algorithm is given to test $1$-juntas. Then~\cite{FKR+04_testing} turned their eyes onto $k$-junta testing problem and gave an $\tilde{O}(k^2)$-queries algorithm. This upper bound was improved by~\cite{Bla09_testing} to a nearly optimal algorithm which requires only $O(k\log k)$ queries, provided an $\Omega(k)$ lower bound proved by~\cite{CG04_lowerbound}. More recently,~\cite{Sag18_lowerbound} gave an $\Omega(k\log k)$ lower bound, which closed the gap. Junta testing has also been investigated in the setting where only non-adaptive queries are allowed~\cite{STW15_adaptivity},~\cite{CST+18_nonadaptive} and~\cite{LCS+19_distributionfree}. Learning $k$-junta boolean function has spawned a large body of work. The learning algorithm obtained by \cite{MOS03_learningjuntas} was a breakthrough and followed by a series of work \cite{LMM+05_learningjuntas}, \cite{AR07_learningjuntawithnoise}, \cite{AM08_learningjuntarandomwalk} and \cite{AKL09_learningjuntaparameterized} discussing $k$-junta learning problem under different circumstances, such as learning symmetric juntas, learning with noise, agnostically learning and considering parameterized learnability. Meanwhile, \cite{BC18_learningjuntas_membershipquery} tried to understand this problem in the membership query model, and \cite{LW18_tolerant_lowerbound}, \cite{BCE+19_tolerant} and \cite{DMN19_tolerant} turned their eyes to tolerant learning $k$-juntas. More recently, people paid more and more attention to learning $k$-junta distribution, see \cite{ABR16_learningjuntadistributions} and \cite{CJL+21_learningjuntadistributions} for more details.

It is expected that a speedup can be obtained when we use a quantum computer to test or learn boolean juntas. In~\cite{AS07_qtesting}, Atici and Servedio gave the first quantum algorithm, which tests $k$-junta boolean functions with $O(k)$ queries. More recently,~\cite{ABRdW16_qtesting} constructed a quantum algorithm which needs only $\tilde{O}(\sqrt{k})$ queries. This was shown to be optimal up to a polylogarithmic factor~\cite{BKT20_lowerbound} In addition,~\cite{AS07_qtesting} also gave an $O\p{2^k}$-sample quantum algorithm for learning $k$-junta boolean function in the PAC model.


In quantum computing, it is natural to consider the situation where behind the oracle is a quantum operation instead of a boolean function.
The quantum junta unitary testing problem is to decide if a unitary $U$ with oracle access is a $k$-junta or $\vep$-far from any $k$-junta unitary.

Wang gave an algorithm testing $k$-junta unitaries with $O(k)$ queries in~\cite{Wan11_qtestingq}. Montanaro and Osborne gave a different tester for dictatorship, i.e., $1$-junta in~\cite{MO10_QBF}. Recently, Chen, Nadimpalli and Yuen have settled both the quantum testing and learning of quantum juntas problem providing nearly tight upper and lower bounds in ~\cite{CNY22_juntatesting}. See Table~\ref{tab:rw} for more details.


The algorithms of testing and learning boolean juntas heavily rely on the Fourier analysis of boolean functions, which is nowadays a rich theory and has wide applications in many branches of theoretical computer science. Readers may refer to~\cite{O14_ABF} or~\cite{De08_intro_to_BF} for more details.
Fourier analysis on quantum operations has received increasing attention in the past couple of years. \cite{MO10_QBF} initiated the study of Fourier analysis on the space of operators and established several interesting properties.
Influence is a key notion in Fourier analysis, which describes how much the function value is affected by some subset of inputs and has many applications in theoretical computer science. The analogous notion in the space of operators has also played a crucial role in designing testing and learning algorithms of $k$-junta unitaries in~\cite{Wan11_qtestingq},~\cite{CNY22_juntatesting}.

We summarize related works in Table~\ref{tab:rw}.

\begin{table}[ht]\label{tab:rw}
\renewcommand\arraystretch{1.5}
\centering
\caption{Our contributions and prior work on testing and learning boolean and quantum $k$-juntas.}
\begin{tabular}[t]{lccc}
\hline
&Classical Testing & Quantum Testing & Quantum Learning\\
\hline
$f:\{0,1\}^n\rarr \{0,1\}$&\makecell{ $O(k\log k)$ \\~\cite{Bla09_testing} } & \makecell{$\tilde{O}(\sqrt{k})$ \\\cite{ABRdW16_qtesting} } & \makecell{ $O(2^k)$ \\~\cite{AS07_qtesting} }\\
&\makecell{ $\Omega(k\log k)$ \\~\cite{Sag18_lowerbound} } & \makecell{$\Omega(\sqrt{k})$ \\\cite{BKT20_lowerbound} } & \makecell{ $\Omega(2^k)$ \\~\cite{AS07_qtesting} }\\
Unitary $U\in \mathcal{M}_{2^n\times 2^n}$&--- & \makecell{$\tilde{O}(\sqrt{k})$ \\\cite{CNY22_juntatesting} } & \makecell{ $O(4^k)$ \\~\cite{CNY22_juntatesting} }\\
&--- & \makecell{$\Omega(\sqrt{k})$ \\~\cite{CNY22_juntatesting} } & \makecell{ $\Omega(4^k/k)$ \\~\cite{CNY22_juntatesting} }\\
Channel $\Phi$, $n$ to $n$ qubits &--- & \makecell{$O(k)$ \\ \textbf{this work} } & \makecell{ $O(4^k)$ \\ \textbf{this work} }\\
&--- & \makecell{$\Omega(\sqrt{k})$ \\ \textbf{this work} } & \makecell{ $\Omega(4^k/k)$ \\ \textbf{this work} }\\\hline
\end{tabular}
\end{table}%

\subsection{Techniques}\label{sec:tech}

In this section, we give a high-level technical overview of our main results.

\subsubsection{Testing Junta Channels}

Our junta testing algorithm is inspired by the algorithm for $k$-junta boolean function testing by~\cite{AS07_qtesting}. The algorithm deeply relies on the notion of the influence of superoperators, which captures how much a subset of input qubits affect the output of a channel; see Section~\ref{sec:inf} for more details. The influence of a superoperator is defined through the formal Fourier analysis framework over superoperators. We prove in Section~\ref{sec:inf} that it has many properties similar to the influence of boolean functions and unitaries.

To prove the lower bound, we reduce $k$-junta channel testing to $k$-junta boolean function testing, which has a lower bound $\Omega(\sqrt{k})$ by~\cite{BKT20_lowerbound}. To make the reduction work, we prove that a tester for $k$-junta channels is also a tester of $k$-junta boolean function, if we view a boolean function as a quantum channel. Moreover, we also show that our algorithm naturally induces a tester for $k$-junta unitaries.

\subsubsection{Learning Junta Channels}

The learning algorithm is inspired by the  algorithms in~\cite{AS07_qtesting,CNY22_juntatesting}. We apply \algname{Pauli-Sample} to the Choi-state of the channel to find the high-influence registers. Then we apply the efficient quantum state tomography algorithm by~\cite{OW17_tomo} to learn the reduced density operator on the qubits with high influence. The lower bound of learning $k$-junta channels is obtained by reducing learning $k$-junta unitaries to learning $k$-junta channels.

\subsubsection{Influence Estimator}

We propose a new Influence-Estimator to estimate influence for channels, which only needs single-qubit operations. To achieve this, we utilize the ideas from the CSS code and quantum money. The estimator uses Hadamard operators to exchange bit-flip effects and phase-flip effects imposed by the channel and finally decides whether the channel changes the target subset of input qubits too much; see Section~\ref{app:ie}.


\section{Preliminary}\label{sec:pre}

We assume that readers are familiar with elementary quantum computing and information theory. Readers may refer to Chapters 1 and 2 of~\cite{NC00_qc} and Chapters 1 and 2 of~\cite{Wat18_TQI} for more detailed backgrounds. For natural number $n\ge 1$,  $[n]$ represents $\{1,2,\dots,n\}$. $I_n$ represents an $n\times n$ identity matrix. The subscript may be omitted whenever it is clear from the context. We say a Hermitian matrix is a positive semidefinite matrix (PSD) if all the eigenvalues are nonnegative.

Throughout the paper, we assume that the whole quantum system has $n$ qubits. Let $N=2^n$ be the dimension of the system. Denote $\Sigma=[N]$ $X=\CC^\Sigma$. $L(X)$ represents the set of all the linear maps from $X$ to $X$ itself. Therefore all $n$-qubit quantum states are a subset of $L(X)$. We note that $L(X)$ is isomorphic to $\CC^{N\times N}$, the set of $N\times N$ matrices. For any $A\in L(X)$, let $\vect{A}=(A\otimes I)\sum_{i=1}^N\ket{i,i}$ be the ``stretching'' column vector of $A$. For $x\in \Zfn$ and $T\subseteq[n]$, let $x_T\in \ZZ_4^T$ be the substring of $x$ obtained by restricting $x$ to all the coordinates in $T$ . We write $0^T\in \ZZ^T_4$ to denote all zero string on coordinates in $T$. The superscript may be dropped whenever it is clear from the context. We use $A^*$ to stand for the conjugate transpose of $A$.

Recall the definition of Pauli operators given by
$$
    \sigma_0=\p{\begin{matrix}1&0\\0&1\end{matrix}}=I,\quad
    \sigma_1=\p{\begin{matrix}0&1\\1&0\end{matrix}}=X,\quad
    \sigma_2=\p{\begin{matrix}0&-i\\i&0\end{matrix}}=Y,\quad
    \sigma_3=\p{\begin{matrix}1&0\\0&-1\end{matrix}}=Z.
$$
It forms an orthogonal basis for $L(\CC^2)$ (over \CC) with respect to the Hilbert-Schmidt inner product. For any $x\in\Zfn$, let $\sigma_x=\otimes_{i=1}^n\sigma_{x_i}$. It is easy to check $\lrs{\sigma_x}_{x\in \Zfn}$ is an orthogonal basis for $L(\CC^{N})=L(X)$.

For $x\in\Zfn$, $\ket{v(\sigma_x)}$ represents the quantum state corresponding to column vector $\vect{\sigma_x}$. It is easy to check $\{\ket{v(\sigma_x)}\}_{x\in \Zfn}$ is an orthogonal basis in $\mathbb{C}^{2^{2n}}$.

\subsection{Superoperators and quantum channels}\label{sec:superoperators}

A superoperator on $L(X)$ is a linear map from $L(X)$ to itself. $T(X)$ represents the set of all superoperators on $L(X)$. A quantum channel $\Phi: L(X) \rarr L(X)$ is completely positive and a trace preserving superoperator. In this work, we concern ourselves with the channels mapping $n$ qubits to $n$ qubits. We use $C(X)$ to denote the set of all quantum channels from $L(X)$ to itself. For any unitary $U\in L(X)$, $\Phi_U$ represents the channel which acts $U$ on the state, i.e., $\Phi_U(\rho)=U\rho U^*$. For any boolean function $g:\{0,1\}^n\rightarrow\{0,1\}$, we define $\Phi_g=\Phi_{U_g}$, where $U_g$ is the unitary defined to be $U_g\ket{x}=(-1)^{g(x)}\ket{x}$ for $x\in\{0,1\}^n$. Next, we introduce the Kraus representation and the Choi representation of superoperators. The properties and relations around two representations are postponed to Appendix~\ref{app:KCR}.

\begin{definition}[Kraus representations, Choi Representations, Choi states]
    Given superoperator $\Phi\in T(X)$, its Kraus representation is
    $$
        \Phi(\rho)=\sum_{s\in \Sgm} A_s\rho B_s^*
    $$
    where $A_s, B_s\in L(X)$. Its Choi representation is
    $$
        J(\Phi) = \sum_{a,b\in \Sgm} \Phi(\ketbra{a}{b})\otimes \ketbra{a}{b} = \p{\Phi\otimes I}\p{\sum_{a,b\in \Sgm}\ketbra{a}{b}\otimes \ketbra{a}{b}} \in L(X\otimes X),
    $$
    where $J$ is a linear map from $T(X)$ to $L(X \otimes X)$.

    For a quantum channel $\Phi$, the Choi state $v(\Phi)$ is defined to be
    \[v(\Phi)=\frac{J(\Phi)}{\Tr J(\Phi)}.\]
    The Choi state of unitaries is defined similarly. Note that for a unitary $U$, its Choi state is a pure state, denoted by $\ket{v(U)}$.
\end{definition}
By Fact~\ref{fact:property_KCP}, $v(\Phi)$ is a density operator if $\Phi$ is a quantum channel.

At the end of this section we introduce $k$-junta channels.

\begin{definition}[$k$-Junta Channels] Given $\Phi \in C(X)$ and a subset $T\subseteq[n]$, we say $\Phi$ is a $T$-junta channel if $\Phi = \Phi_T \otimes I_{T^c}$. $\Phi$ is a $k$-junta channel if $\Phi$ is a $T$-junta channel for some $T\subseteq[n]$ of size $k$.
\end{definition}


\section{Fourier Analysis over superoperators}\label{sec:fourier}

We are ready to introduce the Fourier analysis over superoperators. For any superoperators $\Phi,\Psi\in T(X)$, define the inner product $\ip{\Phi, \Psi} = \ip{J(\Phi), J(\Psi)} = \Tr J(\Phi)^*J(\Psi)$. It is easy to verify that $\ip{\cdot,\cdot}$ is an inner product and  $(T(X),\ip{\cdot,\cdot})$ forms a finite-dimensional Hilbert space. The norm of $\Phi$ is defined to be $\norm{\Phi} = \sqrt{\ip{\Phi,\Phi}} = \norm{J(\Phi)}_2$, where $\|\cdot\|_2$ is the Frobenius norm. The distance between
$\Phi$ and $\Psi$ is defined to be

\begin{equation}\label{def:distance}
D(\Phi,\Psi)=\f{1}{N\sqrt{2}}\norm{\Phi - \Psi}= \f{1}{N\sqrt{2}}\norm{J(\Phi) - J(\Psi)}_2
\end{equation}
The normalizer $N\sqrt{2}$ simply keeps the distance between two quantum channels in $[0,1]$.

Provided the definitions above, we are going to introduce an orthogonal basis.

\begin{definition}[Orthogonal Basis for Superoperators]
	For any $x,y\in\Zfn$, let
	\begin{equation}\label{def of SOB}
	    \Phi_{x,y}(\rho)=\sigma_x\rho\sigma_y.
	\end{equation}
\end{definition}

\begin{restatable}{proposition}{prowd}\label{pro:wd}
	$\{\Phi_{x,y}\}_{x,y\in \ZZ_4^n}$ forms an orthogonal basis in $(T(X),\ip{\cdot,\cdot})$. Besides, $\norm{\Phi_{x,y}}=N$ for all $x,y\in \Zfn$.
\end{restatable}

The proof is deferred to Appendix~\ref{app:wd}. We are ready to define the Fourier expansions of superoperators now.

\begin{definition}[Fourier Expansion of Superoperators]
	For superoperator $\Phi\in T(X)$, the Fourier expansion of $\Phi$ is defined to be
	$$
	    \Phi = \sum_{x,y\in \ZZ_4^n} \wPhi(x,y)\Phi_{x,y}
	$$
	where $\Phi_{x,y}$ is defined by Eq.~\eqref{def of SOB}. $\wPhi(x,y)$'s are the Fourier coefficients of $\Phi$ and  $\wPhi(x,y)=\f{1}{N^2}\ip{\Phi_{x,y}, \Phi}$. Moreover, we define $\wPhi$ to be the $N^2\times N^2$ matrix with entries $\p{\wPhi(x,y)}_{x,y\in \ZZ_4^n}$.
\end{definition}


\begin{restatable}{lemma}{lemmaexistsU}\label{lemma:existsU}
	There exists unitary U such that $\wPhi = \f{1}{N}U^*J(\Phi) U$. Therefore, $\wPhi$ is PSD if and only if $J(\Phi)$ is PSD. In particular, if $J(\Phi)$ is PSD, then $\wPhi(x,x)\in \RR$ for all $x\in \ZZ_4^n$. For a quantum channel $\Phi$, we have $0\le \wPhi(x,x)\le 1$ for all $x\in \ZZ_4^n$ and $\sum_{x\in \ZZ_4^n}\wPhi(x,x)=\Tr \wPhi=1.$
\end{restatable}

The proof is deferred to Appendix~\ref{app:expansion}.


\subsection{Influence}\label{sec:inf}

Given superoperator $\Phi\in T(X)$ and a subset $S\subseteq[n]$, the influence of $S$ on $\Phi$ measures how much the qubits in $S$ affect $\Phi$. It is an extension of the influence on operators introduced by~\cite{MO10_QBF}, which, in turn, is inspired by the analogous notion for boolean functions. We will establish several properties of the influence on quantum channels, which enable us to design both testing algorithms and learning algorithms for $k$-junta channels.

\begin{definition}[Influence of superoperators]
    Given superoperator $\Phi\in T(X)$, $S\subseteq [n]$, the influence of $\Phi$ on $S$ is defined as
    $$
        \inf_S[\Phi] = \sum_{x\in \ZZ_4^n; x_S \neq \0 } \wPhi(x,x).
    $$
    We use $\inf_i[\Phi]$ to represent $\inf_{\{i\}}[\Phi]$ for convenience.
\end{definition}

Notice that the influence of a superoperator can be negative, which is different from operators in~\cite{MO10_QBF} or boolean functions. However, we only concern ourselves about completely positive superoperators, whose influence is always nonnegative by Lemma~\ref{lemma:existsU}.

The following proposition follows from Lemma~\ref{lemma:existsU} directly.

\begin{proposition}
    Given quantum channel $\Phi\in C(X)$,  $S\subseteq [n]$, it holds that $0\leq \inf_S[\Phi]\leq 1$.
\end{proposition}


The following are some basic properties of influence, which can be easily derived from the definition and Lemma~\ref{lemma:existsU}.
\begin{proposition}
    Given quantum channel $\Phi\in C(X)$ and $S,T\subseteq [n]$, we have
    \begin{enumerate}
    	\item $S\subseteq T \Rarr \inf_S[\Phi] \le \inf_T[\Phi]$;
    	\item $\inf_S[\Phi] + \inf_T[\Phi] \ge \inf_{S\cup T}[\Phi]$;
    	\item $\inf_{\varnothing}[\Phi]=0$, $\inf_{[n]}[\Phi] = 1	$.
    \end{enumerate}
\end{proposition}


The following key theorem states that the closeness between a quantum channel and juntas is captured by the influence.

\begin{restatable}[Influence and Distance from k-Junta Channels]{theorem}{thmkeythm}\label{thm:keythm}
	Let $\Phi\in C(X)$ be a quantum channel. If there exists a subset $T\subseteq[n]$ satisfying that $\inf_{T^c}[\Phi]\le \vep$ for $0\leq \vep<1$, then there exists a
$T$-junta channel $\Phi''$ such that $\D(\Phi,\Phi'')\le \sqrt{\vep} + \vep/\sqrt{2}$.
\end{restatable}

To obtain this theorem, we construct a $T$-junta channel $\Phi''$ explicitly from $\Phi$ by two steps. Firstly we construct a $T$-junta ``sub-channel'' $\Phi'$ and then complement it into a $T$-junta channel $\Phi''$. The proof of Theorem~\ref{thm:keythm} is deferred to Appendix~\ref{app:key}.

\begin{corollary}\label{corr:keythm}
	Given quantum channel $\Phi$, if $\Phi$ is $\vep$-far from any $k$-junta channels, then $\inf_{T^c}[\Phi] \ge \vep^2/4$ for all $T\subseteq [n]$ with $\abs{T}\le k$.
\end{corollary}

\subsection{Characterizations of Distance Function}
In this section, we will compare the distance given in Eq.~\eqref{def:distance} with other metrics measuring the distances between two quantum channels.
All the proofs in this section can be found in Appendix~\ref{app:char}.


\cite{CNY22_juntatesting} introduced a distance  $\dist(\cdot,\cdot)$ between unitaries, with which the authors gave optimal testing and learning algorithms for $k$-junta unitaries. The distance $\dist(\cdot,\cdot)$ is defined as follows.

\begin{equation}\label{eq:dU}
  \dist(U,V)=\frac{1}{\sqrt{2N}}\min_{\theta\in[0,2\pi)}\|U-e^{i\theta}V\|_2
\end{equation}

The following lemma asserts that the distance $D(\cdot,\cdot)$ in Eq.~\eqref{def:distance} and $\dist(\cdot,\cdot)$ in Eq.~\eqref{eq:dU} are equivalent when considering unitary operations. Recall that $\Phi_U$ is defined in section~\ref{sec:superoperators}.

\begin{restatable}[Related to distance between Unitaries]{lemma}{lemmareldis}\label{lemma:reldis}
   For unitary matrices $U$ and $V$, it holds that
   $$
       \dist(U,V)\le D(\Phi_U,\Phi_V)\le \sqrt{2} \cdot \dist(U,V).
   $$
\end{restatable}

The following proposition proves that $D(\cdot,\cdot)$ captures the average operator distance between two channels. We expect that our distance function could be used in other channel property testing problems.

%
\begin{restatable}[Related to average-case operator distance]{proposition}{prorelave}\label{pro:relave}
    For quantum channels $\Phi$ and $\Psi$, it holds that
    $$
    \int_{\psi}\norm{\Phi(\ketbra{\psi})-\Psi(\ketbra{\psi})}_2^2d\psi = \f{2N}{N+1}D(\Phi,\Psi)^2+\f{1}{N(N+1)}\norm{\Phi(I)-\Psi(I)}_2^2,
    $$
    where the integral is taken over the Haar measure on all the unit vectors $\psi$.

    Especially for unital channels $\Phi$ and $\Psi$, i.e., $\Phi(I)=\Psi(I)=I$, we have
    $$
    \int_{\psi}\norm{\Phi(\ketbra{\psi})-\Psi(\ketbra{\psi})}_2^2d\psi = \f{2N}{N+1}D(\Phi,\Psi)^2
    $$
\end{restatable}



Similar properties have been established for the distance $\dist(\cdot,\cdot)$ between two unitaries in Proposition 21 of~\cite{MdW16_survey}. We refer interested readers to the discussion about the reason for the chosen distances in Section 5.1.1 of~\cite{MdW16_survey}


Finally, we prove that $D(\cdot,\cdot)$ can be very far from the worst-case operator norm. Here we consider the $1$ to $2$ diamond norm.

\begin{definition}[$1$ to $2$ Diamond Norm]\label{def:dotnorm}
    Given $\Phi\in T(X)$, its $1$ to $2$ diamond norm is defined to be
    $$
        \dotnorm{\Phi}=\otnorm{\Phi\otimes \mathbbm{1}_X}= \max_{\rho:\norm{\rho}_1=1}\{\norm{(\Phi\otimes \mathbbm{1}_X)(\rho)}_2\}
    $$
\end{definition}

\begin{restatable}[Related to worst-case operator distance]{proposition}{prorelworst}\label{pro:relworst}
    For quantum channels $\Phi$ and $\Psi$, it holds that
    $$
    \sqrt{2}D(\Phi,\Psi) \le \dotnorm{\Phi-\Psi} \le N\cdot \sqrt{2}D(\Phi,\Psi)
    $$
Both equalities above can be achieved.
\end{restatable}

\section{Testing $k$-Junta Quantum Channels}\label{sec:testing}

In this section, we show an $O(k)$-query $k$-junta channel testing algorithm and an $\Omega(\sqrt{k})$ lower bound. First, we prove an upper bound on the sample complexity by presenting a $k$-junta channel tester, where the analysis of the algorithm relies on the Fourier analysis of superoperators. The lower bound is obtained by reducing $k$-junta boolean function testing to $k$-junta channel testing. Finally, we show that the $k$-junta channels testing problem is the natural extension of $k$-junta unitary testing problem under our distance function of channels, which gives an alternative proof of the lower bound.


\subsection{$O(k)$ Upper Bound and $\Omega(\sqrt{k})$ Lower Bound}\label{sec:testuplo}

We firstly show our $k$-junta channel tester. 
Our tester is inspired by~\cite{AS07_qtesting}.
However, we introduce a new subroutine, \algname{Influence-Sample} to replace 
    \algname{Pauli-Sample} used in~\cite{AS07_qtesting}.
Note that \algname{Pauli-Sample} requires preparing maximally entanglement state
    and measurement operations over the two-qubit Bell basis,
    whereas \algname{Influence-Sample} requires only single-qubit operations.
In addition, \algname{Influence-Sample} results in only a constant-factor
    decrease in efficiency. 
Our learning algorithm Algorithm~\ref{alg:learner} also includes 
    \algname{Influence-Sample} as a subroutine instead of \algname{Pauli-Sample}.

\begin{myalg}
\caption{\algname{Influence-Sample}($\Phi, t$)}
\label{alg:is}
\Input{Oracle access to quantum channel $\Phi\in C(X)$, a natural number t}	
\Output{$S\subseteq[n]$}

\begin{algorithmic}[1]
\item Initialize $S=\varnothing$;
\item Repeat the following for $t$ times;
\vspace{-\topsep}
\begin{itemize}
\setlength\itemsep{0mm}
\item Uniformly randomly choose $i\in \{0,1\}^n$. 
    Prepare state $\ket{i}$;

\item Uniformly randomly choose $U$ from 
$$\lrs{
    I=\p{\begin{matrix}1&0\\0&1\end{matrix}},
    H=\f{1}{\sqrt{2}}\p{\begin{matrix}1&1\\1&-1\end{matrix}},
    R_x\p{\f{\pi}{2}}=\f{1}{\sqrt{2}}\p{\begin{matrix}1&-i\\-i&1\end{matrix}}
};$$

\item Query $\Phi$ to obtain 
$\p{U^{\otimes n}}^\dagger
    \Phi\p{
        U^{\otimes n}
        \ketbra{i}
        \p{U^{\otimes n}}^\dagger
    }
U^{\otimes n}$;

\item Measure the qubits over the computational basis.
    Let the result be $i'$;

\item $S\leftarrow S\cup \{l\in [n]\mid i_l\neq i'_l\}$.
\end{itemize}
\vspace{-\topsep}
\item Return $S$.
\end{algorithmic}
\end{myalg}

\begin{myalg}
\caption{\algname{Junta-Channel-Tester}($\Phi, k, \vep$)}
\label{alg:tester}
\Input{Oracle access to quantum channel $\Phi\in C(X)$, $k$, $\vep$}	
\Output{``Yes'' or ``No''}
\begin{algorithmic}[1]
\item Let $S=\algname{Influence-Sample}(\Phi, 60(k+1) / \vep^2)$;
\item Output ``Yes'' if $|S|\le k$, or else output ``No''.
\end{algorithmic}
\end{myalg}

\begin{restatable}[Property of Algorithm~\ref{alg:tester}, Restatement of Theorem~\ref{thm:main1}]{theorem}{thmtestuptwo}\label{thm:testup2}
    Given quantum channel $\Phi\in C(X)$, 
        with probability at least $9/10$, 
        the algorithm $\algname{Junta-Channel-Tester}(\Phi,k,\vep)$ 
        outputs ``Yes'' if $\Phi$ is a $k$-junta, 
        and outputs ``No'' if $\Phi$ is $\vep$-far from any $k$-junta channel. 
    The algorithm makes $O\p{k/\vep^2}$ queries to the channel $\Phi$.
\end{restatable}

An algorithm is a $(k,\vep)$-channel junta tester if it can distinguish whether the given channel is $k$-junta or is $\vep$-far from any $k$-junta channels. $(k,\vep)$-classical junta testers and $(k,\vep)$-unitary junta testers are defined similarly.

\begin{restatable}{lemma}{lemmalored}\label{lemma:lored}
    A $(k,\sqrt{\vep/2})$-channel junta tester is a $(k,\vep)$-classical junta tester.
\end{restatable}

Combining Lemma~\ref{lemma:lored} with the $\Omega(\sqrt{k})$ lower bound on testing $k$-junta boolean function proved by~\cite{BKT20_lowerbound}, we obtain an $\Omega(\sqrt{k})$ lower bound on testing $k$-junta channels. Our key technical lemma is as follows. Recall that $\Phi_g$ is defined in Section~\ref{sec:superoperators} for boolean function $g$.

\begin{restatable}{lemma}{lemmalokey}\label{lemma:lokey}
    For a $k$-junta channel $\Phi$, there exists a $k$-junta boolean function $g'$ satisfying that $D(\Phi,\Phi_{g'})=\min_{g} D(\Phi,\Phi_g)$, where the minimization is over all boolean functions $g:\{0,1\}^n\rightarrow\{0,1\}$.
\end{restatable}

With the assistance of this result around the distance structure of $k$-junta channels, we obtain the desired reduction in Lemma~\ref{lemma:lored}. See Appendix~\ref{app:testing} for the detailed proofs.

\subsection{Reduction from $k$-Junta Unitary Testing}\label{sec:Ured}

To show our distance function induced by Fourier analysis over superoperators is a natural extension of the distance function on unitaries discussed in~\cite{MdW16_survey}, we provide an extra reduction from $k$-junta unitary testing. It gives an alternative proof of our testing lower bound. All the proofs can be found in Appendix~\ref{app:testing}.

\begin{restatable}[Reduction from Testing $k$-Junta Unitaries to Testing $k$-Junta Channels]{lemma}{lemmaredtwo}\label{lemma:red2}
   	A $(k,\vep)$-channel junta tester is naturally a $(k,\vep/2)$-unitary junta tester.
\end{restatable}

The key technical result is as follows:


\begin{restatable}{lemma}{lemmaredkey}\label{lemma:redkey}
	For every $k$-junta channel $\Phi'$, there exists a $k$-junta unitary $V$, such that $D(\Phi',\Phi_V)=\min_{V} D(\Phi',\Phi_V)$, where the minimization is over all unitaries $V$.
\end{restatable}

\section{Learning $k$-Junta Quantum Channels}\label{sec:learning}

In this section, we prove a nearly tight bound on $k$-junta learning problem. 
Our algorithm is inspired by the learning algorithms in~\cite{AS07_qtesting} 
    and~\cite{CNY22_juntatesting}. 
As mentioned above, we use \algname{Influence-Sample} to replace 
    \algname{Pauli-Sample} to reduce operation complexity,
    and this replacement results in only a constant-factor decrease in efficiency.
We describe the algorithm \algname{Junta-Channel-Learner} as follows.

\begin{myalg}
\caption{\algname{Junta-Channel-Learner}($\Phi, k, \vep$)}
\label{alg:learner}
\Input{Oracle access to $k$-junta channel $\Phi\in C(X)$, $\vep$}	
\Output{A classical description of $\Phi$ in the form of its Choi representation, a $4^n\times 4^n$ matrix}
\begin{algorithmic}[1]
\item Let $S=\algname{Influence-Sample}(\Phi, O(k\log k / \vep^2))$;
\item Set $t=O(4^k/\vep^2)$. Call $\algname{Quantum-State-Preparation}(\Phi, S)$ for $10t$ times to obtain at least $t$ copies of quantum state $\psi_S$;
\item Return $\algname{Channel-Tomography}(\psi_S^{\otimes t}, \vep)\otimes v(I^{S^c})$ as the result.
\end{algorithmic}
\end{myalg}

\begin{myalg}
\caption{\algname{Quantum-State-Preparation}($\Phi, S\subseteq[n]$)}
\label{alg:qsp}
\Input{Oracle access to $k$-junta channel $\Phi\in C(X)$, $\gamma$}	
\Output{A $2|S|$-qubit quantum state, or ``error''}
\begin{algorithmic}[1]
\item Prepare the state $v(\Phi)$;
\item Measure $2\abs{S^c}$ qubits in $S^c$ onto
 the Pauli basis $\{\ket{\sigma_x}\}_{x\in \ZZ_4^{\abs{S^c}}}$;
\item If the measurement result is $\0^{S^c}$, return the untouched $2\abs{S}$ qubits. Otherwise, return ``error''.
\end{algorithmic}
\end{myalg}

\begin{myalg}
\caption{\algname{Channel-Tomography}($\psi^{\otimes O(4^k/\vep^2)}$, $\vep$)}
\label{alg:ct}
\Input{Independent copies of $\psi$ and $\vep$, enough for $\algname{Tomography}$}	
\Output{A classical description of $\psi$}
\begin{algorithmic}[1]
\item Run $\algname{Tomography}(\psi^{\otimes O(4^k/\vep^2)}, 0.04\vep)$ to obtain a description of state $\psi$;
\item Find out, by only local calculation, the Choi state closest to $\psi$ and return the description.
\end{algorithmic}
\end{myalg}

\begin{restatable}[Property of Algorithm~\ref{alg:learner}, Restatement of Theorem~\ref{thm:main1}]{theorem}{thmplearner}\label{thm:plearner}
	Given oracle access to $k$-junta channel $\Phi$, with probability at least $9/10$, $\algname{Junta-Channel-Learner}(\Phi,k,\vep)$ outputs the description of quantum channel $\Psi$ such that $D(\Phi,\Psi)\le \vep$. Furthermore, this algorithm makes $O(4^k/\vep^2)$ queries.
\end{restatable}

As for the $k$-junta channel learning lower bound, recall Lemma~\ref{lemma:reldis} shows that our distance function over channels is equivalent to the distance between unitaries used in~\cite{CNY22_juntatesting}, up to a constant factor, it is very natural to reduce learning $k$-junta unitaries to learning $k$-junta channels, and therefore the following lower bound follows.

\begin{theorem}[Lower Bound on Learning $k$-Junta Channels]\label{thm:ll}
	Any algorithm learning $k$-junta channels within precision $\vep$ under $D(\cdot,\cdot)$ requires $\Omega(4^k\log(1/\vep)/k)$ queries.
\end{theorem}

\section{Conclusion}\label{sec:summary}

We exhibit two algorithms, one for testing $k$-junta channels and one for learning $k$-junta channels and lower bounds respectively. The $k$-junta channel learning algorithm is nearly optimal. Our algorithms generalize the work~\cite{AS07_qtesting,CNY22_juntatesting} about testing and learning $k$-junta unitaries and $k$-junta boolean function. To design the algorithms and prove the lower bounds, we introduce the Fourier analysis over the space of superoperators, which extends the Fourier analysis over operators in~\cite{MO10_QBF}. As~\cite{MdW16_survey} mentioned, there was not much work on testing the properties of quantum channels. We expect more applications in designing algorithms for testing and learning quantum channels through the Fourier analysis presented in this paper.

\section*{Acknowledgement}
We would thank Jingquan Luo for pointing out 
    an error in the previous version of this work. 
We thank the anonymous reviewers for their careful 
    reading and helpful comments. 
We also would like to express our gratitude for the insightful discussions 
    with Eric Blais and Nengkun Yu.
The first author would also like to extend special thanks to Minglong Qing, 
    Mingnan Zhao and Haochen Xu for their invaluable support in 
    problem-solving and the writing of this paper. 
This work was supported by National Natural Science Foundation of China 
    (Grant No. 62332009, 61972191) and Innovation Program for Quantum Science and
    Technology (Grant No. 2021ZD0302900).

\bibliographystyle{alpha}
\bibliography{bib/juntabib}

\appendix

\section{Properties on Kraus and Choi Representations}\label{app:KCR}

In this section, we list some basic properties of Kraus and Choi representations, whose proofs can be found in Section 2.2 in~\cite{Wat18_TQI}.

\begin{fact}\label{fact:property_KCP} Given superoperators $\Phi\in T(X),\Phi'\in T(X')$, it holds that
\begin{enumerate}
	\item $\Phi$ is completely positive if and only if it has a Kraus representation $\Phi(\rho)=\sum_{s\in \Sgm} A_s\rho A_s^*$. It is trace preserving if and only if its Kraus representation $\Phi(\rho)=\sum_{s\in \Sgm} A_s\rho B_s^*$. satisfies that $\sum_{s\in \Sgm}B_s^*A_s=I$;
	\item  $\Phi$ is completely positive if and only if $J(\Phi)$ is PSD. It is trace preserving if and only if $\Tr_{X_1} J(\Phi)=I_{X_2}$, where $J$ is viewed as a map from $L(X)$ to $L(X_1)\otimes L(X_2)$ with $X_1=X_2=X$;
	\item If $\Phi(\rho)=\sum_{s\in \Sgm} A_s\rho B_s^*$, we have
        $$
            J(\Phi) = \sum_{s\in \Sgm} \vect{A_s}\vect{B_s}^*;
        $$
    \item  $J(\Phi\otimes\Phi')=J(\Phi)\otimes J(\Phi')$.
\end{enumerate}
	
\end{fact}

\section{Fourier Basis of Superoperators Is Well-defined}\label{app:wd}

Here we list some basic properties of the inner product and the norm introduced in Section~\ref{sec:fourier}, which are easy to verify by the definitions.
\begin{fact}[Properties of Inner Product and Norm]\label{fact:property_ipn}
\begin{enumerate}
	\item Given $\Phi(\rho)=A\rho B^*$ and $\Psi(\rho)=C\rho D^*$, we have $\ip{\Phi,\Psi} = \ip{A,C}\cdot \ip{D,B}$;
	\item For $\Phi(\rho)=A\rho B^*$, we have $\norm{\Phi}=\sqrt{\ip{\Phi,\Phi}}=\sqrt{\ip{A,A}\cdot \ip{B,B}}=\norm{A}_2\cdot \norm{B}_2$;
	\item Suppose $\Phi = \Phi_1\otimes \Phi_2$ and $\Psi = \Psi_1 \otimes \Psi_2$. We have $\ip{\Phi,\Psi}=\ip{\Phi_1, \Psi_1}\cdot \ip{\Phi_2, \Psi_2}$.
\end{enumerate}
	
\end{fact}

We are going to prove Proposition~\ref{pro:wd} now.

\prowd*

\begin{proof}
    \par{\emph{Norm}}. $\forall x,y\in \ZZ_4^n$, $\norm{\Phi_{x,y}}=\norm{\sigma_x}_2\norm{\sigma_y}_2 = N$ using Fact~\ref{fact:property_ipn}.
    \par{\emph{Orthogonality}}. $\forall x,x',y,y'\in \ZZ_4^n$,$x\neq x'$ or $y\neq y'$, we have
    \begin{align*}
        \ip{\Phi_{x,y},\Phi_{x',y'}} &= \prod_{i\in[n]} \ip{\Phi_{x_i,y_i},\Phi_{x'_i,y'_i}} \\
    	&= \prod_{i\in[n]} \ip{\sigma_{x_i},\sigma_{x'_i}}\cdot \ip{\sigma_{y'_i},\sigma_{y_i}} \\
    	&= 0
    \end{align*}
    All equalities follow from Fact~\ref{fact:property_ipn} directly. Note that for non-zero vectors, orthogonality implies linear independence.

    \par{\emph{Basis, spanning the whole space}}. The dimension of $T(X)$ is $N^4=2^{4n}$ and we have $4^{2n}=2^{4n}$ linearly independent vectors in $\{\Phi_{x,y}\}_{x,y\in \ZZ_4^n}$.

\end{proof}

\section{Properties of Fourier Expansions of Superoperators}\label{app:expansion}

\lemmaexistsU*

\begin{proof}
    By the definition of $\wPhi$, we have
    \begin{align*}
        \wPhi(x,y)&=\f{1}{N^2}\ip{\Phi_{x,y}, \Phi} \\
        &= \f{1}{N^2}\Tr (\vect{\sigma_x}\vect{\sigma_y}^*)^* J(\Phi) \\
        &= \f{1}{N^2} \vect{\sigma_x}^* J(\Phi) \vect{\sigma_y},
    \end{align*}
    where the second equality is by the definition of the inner product and the fact that $J(\Phi_{x,y})=\vect{\sigma_x}\vect{\sigma_y}^*$.     Therefore
    $$
        \wPhi = \f{1}{N} U^* J(\Phi) U
    $$
    where $U = [vec(\sigma_x)/\sqrt{N}]_{x\in \ZZ_4^n}$ is a unitary.

\end{proof}

The next corollary follows from the properties of Kraus and Choi representations in Fact~\ref{fact:property_KCP}. We note that $\Phi(\rho) = \sum_{x,y\in \ZZ_4^n} \wPhi(x,y)\sigma_x\rho\sigma_y$ is a Kraus representation of $\Phi$. Therefore $\Phi\in T(X)$ is trace preserving if and only if $\sum_{x,y\in \Zfn}\wPhi(x,y)\sigma_y\sigma_x=I$.

\begin{corollary}\label{corr:channel_check}
    Let $\Phi\in T(X)$ be a superoperator. The following statements are equivalent.
    \begin{enumerate}
      \item $\Phi\in T(X)$ is completely positive.
      \item $\wPhi$ is PSD.
    \end{enumerate}
    The following statements are equivalent as well.
    \begin{enumerate}
      \item $\Phi\in T(X)$ is trace preserving.
      \item $\sum_{x,y\in \Zfn}\wPhi(x,y)\sigma_y\sigma_x=I$.
    \end{enumerate}
\end{corollary}

\begin{corollary}[Relations between Fourier Expansion and Norm and Distance]\label{corr:rel_dis_norm} Let $\Phi,\Psi\in T(X)$ be superoperators and $\wPhi$, $\wPsi$ be the corresponding Fourier expansions. Then
	\begin{enumerate}
		\item $\norm{\Phi}=N\norm{\wPhi}_2=N\sqrt{\sum_{x,y\in \Zfn}\abs{\wPhi(x,y)}^2}$;
		\item $\D(\Phi,\Psi)= \f{1}{\sqrt{2}}\norm{\wPhi-\wPsi}_2=\f{1}{\sqrt{2}}\sqrt{\sum_{x,y\in \Zfn}\abs{\wPhi(x,y) - \wPsi(x,y)}^2}$.
	\end{enumerate}
\end{corollary}

\section{Proof of Theorem~\ref{thm:keythm}}\label{app:key}

\thmkeythm*

\begin{proof}
    We need two steps to construct $\Phi''$ explicitly. Firstly we construct a $k$-junta \textit{sub-channel} $\Phi'$, which is completely positive and trace non-increasing, and then turn it to a channel $\Phi''$.	
    \par{\textbf{Construction of \textit{sub-channel} $\Phi'$}} Let
    $$
        \Phi'(\rho)=\sum_{x,y\in \Zfn; x_{T^c}=y_{T^c}=\0} \wPhi(x,y)\sigma_x\rho\sigma_y
    $$
    Notice that $\Phi$ is a quantum channel. By Fact~\ref{fact:property_ipn} and Corollary~\ref{corr:channel_check}, it is easy to see $\Phi'$ is a $T$-junta sub-channel. Notice that $\wPhi'$ is a principle submatrix of PSD matrix $\wPhi$, which implies $\wPhi'$ is PSD. Then again by Corollary~\ref{corr:channel_check}, $J(\wPhi')$ is also PSD. Now we bound the distance between $\Phi$ and $\Phi'$ from above.

    By Corollary~\ref{corr:rel_dis_norm}, we have
    $$
        2\cdot \D(\Phi,\Phi')^2 = \sum_{x,y\in \Zfn; x_{T^c}\neq \0 \text{ or }y_{T^c}\neq \0}\abs{\wPhi(x,y)}^2
    $$
     For any $x,y\in \Zfn$ we have $\abs{\wPhi(x,y)}^2\le \wPhi(x,x)\wPhi(y,y)$ since $\wPhi$ is a PSD matrix. This implies
    \begin{eqnarray}
      \sum_{x,y\in \Zfn; x_{T^c}\neq \0 \text{ or }y_{T^c}\neq \0}\abs{\wPhi(x,y)}^2 &\le& \sum_{x,y\in \Zfn: x_{T^c}\neq \0 \text{ or }y_{T^c}\neq \0} \wPhi(x,x)\wPhi(y,y) \nonumber\\
      &\le& \left(\sum_{x,y\in \Zfn: x_{T^c}\neq \0}+\sum_{x,y\in \Zfn: y_{T^c}\neq \0}\right)\wPhi(x,x)\wPhi(y,y).\label{eqn:distance}
    \end{eqnarray}

    Notice $\sum_{x\in \Zfn}\wPhi(x,x)=1$ by Lemma~\ref{lemma:existsU}. We have
    $$
        \text{RHS of Eq.~\eqref{eqn:distance}}= 2 \sum_{x\in \Zfn; x_{T^c}\neq \0} \wPhi(x,x)     $$
    To summarize, we have
    $$
    2\cdot \D(\Phi,\Phi')^2 \le 2 \sum_{x\in \Zfn; x_{T^c}\neq \0} \wPhi(x,x)=2\cdot \inf_{T^c}[\Phi]\le 2\vep
    $$

    We claim that $\sum_{x,y\in\Zfn; x_{T^c}=y_{T^c}=\0}\wPhi(x,y)\sigma_y\sigma_x \le I$.

    Let
    \begin{align}
        A &= \sum_{x,y\in \Zfn, x_{T^c}=y_{T^c}=\0} \wPhi(x,y)\sigma_y\sigma_x,\label{eqn:A} \\
        B &= \sum_{x,y\in \Zfn, x_{T^c}\neq \0,y_{T^c}\neq\0} \wPhi(x,y)\sigma_y\sigma_x, \nonumber\\
        C_1 &= \sum_{x,y\in \Zfn, x_{T^c}= \0,y_{T^c}\neq\0} \wPhi(x,y)\sigma_y\sigma_x, \nonumber\\
        C_2 &= \sum_{x,y\in \Zfn, x_{T^c}\neq \0,y_{T^c}=\0} \wPhi(x,y)\sigma_y\sigma_x.\nonumber
    \end{align}

    We note that $A=\sum_{x',y'\in \ZZ_4^{T}} \wPhi(x'\circ 0^{T^c},y'\circ 0^{T^c})\sigma_{y'}\sigma_{x'}\otimes I^{T^c}=:A'\otimes I^{T^c}$, where $x'\circ 0^{T^c}$ is the concatenation of $x'$ and $0^{T^c}$. Same for $y'\circ 0^{T^c}$. To see $A\le I$, it is enough to show $A'\le I$, which is equivalent to $\Tr A'\ketbra{\phi}\le 1$ for any quantum state $\ket{\phi}_T$.

    Let $I_{T^c}$ be a $2^{\abs{T^c}}\times {2^{\abs{T^c}}}$ identity matrix. Notice that $\Tr A(\ketbra{\phi}\otimes \frac{I_{T^c}}{2^{\abs{T^c}}})=\Tr (A'\cdot\ketbra{\phi})$. It suffices to prove that $\Tr A(\ketbra{\phi}\otimes \frac{I_{T^c}}{2^{\abs{T^c}}})\le 1$ for arbitrary quantum state $\ket{\phi}$. To this end,
    \begin{align*}
        1 &= \Tr \Phi\p{\ketbra{\phi}\otimes \frac{I_{T^c}}{2^{\abs{T^c}}} } \\
          &=\frac{1}{2^{\abs{T^c}}}\p{ \Tr A(\ketbra{\phi}\otimes I_{T^c}) + \Tr B(\ketbra{\phi}\otimes I_{T^c}) +\Tr C_1(\ketbra{\phi}\otimes I_{T^c}) +\Tr C_2(\ketbra{\phi}\otimes I_{T^c})} \\
          &\ge \frac{1}{2^{\abs{T^c}}}\p{\Tr A(\ketbra{\phi}\otimes I_{T^c}) + \Tr C_1(\ketbra{\phi}\otimes I_{T^c}) +\Tr C_2(\ketbra{\phi}\otimes I_{T^c})} \\
          &= \frac{1}{2^{\abs{T^c}}}\p{\Tr A(\ketbra{\phi}\otimes I_{T^c})}
    \end{align*}
    where the first inequality is because $B$ is a principle sub-matrix of $\wPhi$, which is also PSD by Corollary~\ref{corr:channel_check}; the last equality is because $\Tr C_1(\ketbra{\phi}\otimes I_{T^c}) +\Tr C_2(\ketbra{\phi}\otimes I_{T^c})=0$. To see this, we will prove that $\Tr C_1(\ketbra{\phi}\otimes I_{T^c})=0$.
    \begin{align*}
        	\Tr C_1(\ketbra{\phi}\otimes I_{T^c})&= \sum_{x,y\in \Zfn, x_{T^c}= \0,y_{T^c}\neq\0} \wPhi(x,y)\Tr \sigma_y\sigma_x(\ketbra{\phi}\otimes I_{T^c}) \\
        	&= \sum_{x,y\in \Zfn, x_{T^c}= \0,y_{T^c}\neq\0}\wPhi(x,y) \bra{\phi}\sigma_{y_T}\sigma_{x_T}\ket{\phi}\cdot \ip{\sigma_{y_{T^c}},\sigma_{x_{T^c}}}\\
        	&= 0
    \end{align*}
    $\Tr C_2(\ketbra{\phi}\otimes I_{T^c})=0$ follows from the same argument. Therefore $A\le I$.

    \par{\textbf{Construction of \textit{Channel} $\Phi''$}} We set
    $$
        \Phi''(\rho)=\Phi'(\rho) + \sqrt{I-A}\rho\sqrt{I-A},
    $$
    where $A$ is given in Eq.~\eqref{eqn:A}.  By Corollary~\ref{corr:channel_check}, we have $J(\Phi'')=J(\Phi') + \vect{\sqrt{I-A}}\vect{\sqrt{I-A}}^*$, which is PSD. Notice that $A=A'\otimes I_{T^c}$. Thus  $\Phi''$ is also a T-junta completely positive map. To prove $\Phi''$ is a channel, it suffices to prove that $\Phi''$ is trace-preserving. By the Kraus representation of $\Phi''$
    $$
        \Phi''(\rho)=\sum_{x,y\in \Zfn; x_{T^c}=y_{T^c}=\0} \wPhi(x,y)\sigma_x\rho\sigma_y + \sqrt{I-A}\rho\sqrt{I-A},
    $$
    we have
    $$
        \sum_{x,y\in \Zfn; x_{T^c}=y_{T^c}=\0} \wPhi(x,y)\sigma_y\sigma_x+\sqrt{I-A}\sqrt{I-A} = A + \sqrt{I-A}\sqrt{I-A} = I,
    $$
    which implies $\Phi''$ is trace preserving according to Fact~\ref{fact:property_KCP}.

    Next we bound the distance between $\Phi'$ and $\Phi''$ from above. Note that $J(\Phi'')=J(\Phi') + \vect{\sqrt{I-A}}\vect{\sqrt{I-A}}^*$, we have
    \begin{eqnarray*}
     D(\Phi'',\Phi')& =& \f{1}{N\sqrt{2}}\norm{J(\Phi'')-J(\Phi')}_2 =	\f{1}{N\sqrt{2}}\norm{\vect{\sqrt{I-A}}\vect{\sqrt{I-A}}^*}_2 \\
      &=& \f{1}{N\sqrt{2}} \norm{\sqrt{I-A}}_2^2 =\f{1}{N\sqrt{2}} \Tr(I-A)
    \end{eqnarray*}
    From the definition of $A$, we have  $\f{1}{N}\Tr A=\sum_{x\in \Zfn, x_{T^c}=\0}\wPhi(x,x)=1-\inf_{T^c}[\Phi]$, which implies
    $$
        \f{1}{N\sqrt{2}}\Tr(I-A) = \f{1}{\sqrt{2}}\inf_{T^c}[\Phi].
    $$
    Therefore
    $$
        D(\Phi'',\Phi') \le \f{1}{\sqrt{2}}\inf_{T^c}[\Phi]\le \f{\vep}{\sqrt{2}}.
    $$
    In conclusion, $\Phi''$ is a $T$-junta channel and $D(\Phi,\Phi'')\le \sqrt{\vep} + \vep/\sqrt{2}$, which completes the proof.

\end{proof}

\section{Characterization of Distance Function}\label{app:char}

\subsection{Proof of Lemma~\ref{lemma:reldis}}

\lemmareldis*

\begin{proof}
       It's easy to see
       $$
           \dist(U,V)=\sqrt{1-\f{1}{N}\abs{\ip{U,V}}}
       $$
       and
       $$
           D(\Phi_U,\Phi_V)=\sqrt{1-\f{1}{N^2}\abs{\ip{U,V}}^2}
       $$
       Let $\alpha=\f{1}{N}\abs{\ip{U,V}}\in[0,1]$.  Lemma~\ref{lemma:reldis} follows from the inequality $\sqrt{1-\alpha}\le \sqrt{1-\alpha^2}\le \sqrt{2}\sqrt{1-\alpha }$.

   \end{proof}

\subsection{Comparison with other operator norms}

\prorelave*

\begin{proof}
    Let $J=J(\Phi)-J(\Psi)=J(\Phi-\Psi)=\sum_{i,j\in[N]}J_{i,j}\otimes \ketbra{i}{j}$, where $J_{i,j}=(\Phi-\Psi)(\ketbra{i}{j})$.
    \begin{align*}
    	\norm{\Phi(\ketbra{\psi})-\Psi(\ketbra{\psi})}_2^2 &= \norm{\Tr_{X_2}\p{J\cdot (I\otimes \ketbra{\psi})}}_2^2 \\
    	&= \norm{\sum_{i,j\in [N]} J_{i,j}\braket{j}{\psi}\braket{\psi}{i}}_2^2 \\
    	&= \sum_{i,j,i',j'\in [N]} \ip{J_{i,j},J_{i',j'}}\braket{i}{\psi}\braket{\psi}{j}\braket{j'}{\psi}\braket{\psi}{i'}
    \end{align*}
    Note that $\braket{i}{\psi}\braket{\psi}{j}\braket{j'}{\psi}\braket{\psi}{i'} = \Tr\ \p{\ketbra{j}{i}\otimes \ketbra{i'}{j'}}\cdot \p{\ketbra{\psi}\otimes \ketbra{\psi}}$, we have
    \begin{align*}
    	&\int_{\psi}\norm{\Phi(\ketbra{\psi})-\Psi(\ketbra{\psi})}_2^2d\psi \\
    	=& \sum_{i,j,i',j'\in [N]} \ip{J_{i,j},J_{i',j'}}\Tr\,\p{\ketbra{j}{i}\otimes \ketbra{i'}{j'}}\cdot \int_{\psi}\ketbra{\psi}^{\otimes2}d\psi \\
    	=& \sum_{i,j,i',j'\in [N]} \ip{J_{i,j},J_{i',j'}}\Tr\,\p{\ketbra{j}{i}\otimes \ketbra{i'}{j'}}\cdot \f{I+F}{N(N+1)} \\
    	=& \f{1}{N(N+1)}\p{\sum_{i,j\in [N]}\ip{J_{i,j},J_{i,j}}+\sum_{i,j\in[N]}\ip{J_{i,i},J_{j,j}}}
    \end{align*}
    In the second equality we use the fact that $\int_{\psi}\ketbra{\psi}^{\otimes2}d\psi=(I+F)/N(N+1)$, where $F$ is the swap operator which interchanges two $n$-qubit quantum systems; see Lemma 7.24 of~\cite{Wat18_TQI}. The third equality follows from $\Tr ((A\otimes B)F)=\Tr AB$.

    By the definition of $J$, we have
    \begin{align*}
    	\norm{J(\Phi)-J(\Psi)}_2^2 &= \norm{J}_2^2=\ip{J,J}=\sum_{i,j\in [N]}\ip{J_{i,j},J_{i,j}} \\
    	\norm{\Phi(I)-\Psi(I)}_2^2&=\norm{\sum_{i\in[N]}J_{i,i}}_2^2=\ip{\sum_{i\in[N]}J_{i,i}, \sum_{j\in[N]}J_{j,j}}=\sum_{i,j\in[N]}\ip{J_{i,i},J_{j,j}}
    \end{align*}
    and therefore
    \begin{align*}
        &\int_{\psi}\norm{\Phi(\ketbra{\psi})-\Psi(\ketbra{\psi})}_2^2d\psi \\
        =&\f{1}{N(N+1)}\p{\sum_{i,j\in [N]}\ip{J_{i,j},J_{i,j}}+\sum_{i,j\in[N]}\ip{J_{i,i},J_{j,j}}} \\
    	=& \f{1}{N(N+1)}\p{\norm{J(\Phi)-J(\Psi)}_2^2+\norm{\Phi(I)-\Psi(I)}_2^2} \\
    	=& \f{2N}{N+1}D(\Phi,\Psi)^2+\f{1}{N(N+1)}\norm{\Phi(I)-\Psi(I)}_2^2
    \end{align*}


\end{proof}

\prorelworst*

\begin{proof}
	We will show $$\f{1}{N}\norm{J(\Phi)-J(\Psi)}_2 \le \otnorm{(\Phi-\Psi)\otimes \mathbbm{1}_X} \le \norm{J(\Phi)-J(\Psi)}_2.$$ Let $\ket{\Pz}=\f{1}{\sqrt{N}}\sum_{i\in X}\ket{ii}$, we have
	$$
	    \f{1}{N}\norm{J(\Phi)-J(\Psi)}_2=\norm{((\Phi-\Psi)\otimes I_X)(\ketbra{\Pz})}_2\le \otnorm{(\Phi-\Psi)\otimes I_X}
	$$
	and the first inequality follows immediately.\
	To prove the next inequality, By the following fact, the $1\rightarrow 2$ diamond norm can be achieved by a rank-$1$ Hermitian matrix.
	\begin{fact}[Theorem 3.51 in~\cite{Wat18_TQI}]\label{fact:univec}
		There exists an unit vector $u\in X\otimes X$, which satisfies that $\dotnorm{\Phi-\Psi}=\norm{((\Phi-\Psi)\otimes I_X)(uu^*)}_2$.
	\end{fact}
	Let $u$ be the unit vector in Fact~\ref{fact:univec} and $A$ be a matrix satisfying that $u=\vect{A}=\sqrt{N}(A\otimes I)\ket{\Pz}=\sqrt{N}(I\otimes A^T)\ket{\Pz}$. We have
	\begin{align*}
	    \dotnorm{\Phi-\Psi} &= \norm{((\Phi-\Psi)\otimes I)(\vect{A}\vect{A}^*)}_2 \\
	    &= N\cdot \norm{((\Phi-\Psi)\otimes I)((I\otimes A^T)\ketbra{\Pz}(I\otimes A^T)^*)}_2 \\
	    &= N\cdot \norm{(I\otimes A^T)((\Phi-\Psi)\otimes I)(\ketbra{\Pz})(I\otimes A^T)^*}_2 \\
	    &= \norm{(I\otimes A^T)(J(\Phi)-J(\Psi))(I\otimes A^T)^*}_2
	\end{align*}
	Applying the norm inequality $\norm{ABC}_2\le \norm{A}_\infty\norm{B}_2\norm{C}_\infty$ and $\norm{I\otimes A^T}_\infty=\norm{A^T}_\infty=\norm{A}_\infty\le \norm{A}_2$, we have
	\begin{align*}
	    \norm{(I\otimes A^T)(J(\Phi)-J(\Psi))(I\otimes A^T)^*}_2&\le \norm{I\otimes A^T}_\infty\cdot \norm{J(\Phi)-J(\Psi)}_2 \cdot \norm{(I\otimes A^T)^*}_\infty \\
	    &\le \norm{A}_2^2\cdot \norm{J(\Phi)-J(\Psi)}_2 \\
	    &= \norm{J(\Phi)-J(\Psi)}_2
	\end{align*}
	The last equality is because $\norm{A}_2=\norm{vec(A)}_2=\norm{u}_2=1$.
	
	To see the tightness of the first inequality, let $\Phi=\Phi_U$ where $U=X\otimes I^{\otimes n-1}$ and $\Phi_U$ is defined in section~\ref{sec:superoperators}. Let $\Psi$ be an identity channel, it's easy to check $\norm{J(\Phi)-J(\Psi)}_2/N = \dotnorm{\Phi-\Psi}=\sqrt{2}$. For second inequality, let $\Phi$ be
	\begin{align*}
	    &\Phi(\ketbra{1})=\ketbra{2}, \Phi(\ketbra{2})=\ketbra{1} \\
	    &\Phi(\ketbra{i})=\ketbra{i}, \forall i\neq 1,2\\
	    &\Phi(\ketbra{i}{j})=0^{N\times N},\forall i\neq j
	\end{align*}
	and $\Psi$ be
	\begin{align*}
	    &\Psi(\ketbra{i})=\ketbra{i}, \forall i\\
	    &\Psi(\ketbra{i}{j})=0^{N\times N},\forall i\neq j
	\end{align*}
	To verify $\Phi$ and $\Psi$ are quantum channels, note that $J(\Phi)$ and $J(\Psi)$ are both PSD and they are trace preserving obviously. Meanwhile, $\norm{J(\Phi)-J(\Psi)}_2=\sqrt{2}=\norm{(\Phi-\Psi)(\ketbra{1})}_2$.
	
\end{proof}

\section{A Simple Influence Estimator}\label{app:ie}

In this section, we will describe a new influence estimator. The estimator only includes single-qubit operations though it fulfills the same function efficiently as the raw influence estimator in~\cite{CNY22_juntatesting}, which needs two-qubit operations and maximally entanglement states.

\begin{myalg}
\caption{\algname{Influence-Estimator}($\Phi, S$)}\label{alg:rie}
\Input{Oracle access to quantum channel $\Phi\in C(X)$, $S\subseteq [n]$}	
\Output{$Y\in \{0,1\}$}

\begin{algorithmic}[1]
    \item Uniformly randomly choose $i\in \{0,1\}^S$, $j\in \{0,1\}^{S^c}$. 
        Prepare state $\ket{i}_S\ket{j}_{S^c}$;

    \item Uniformly randomly choose $U$ from 
        $$\lrs{
            I=\p{\begin{matrix}1&0\\0&1\end{matrix}},
            H=\f{1}{\sqrt{2}}\p{\begin{matrix}1&1\\1&-1\end{matrix}},
            R_x\p{\f{\pi}{2}}=\f{1}{\sqrt{2}}\p{\begin{matrix}1&-i\\-i&1\end{matrix}}
        };$$
    
    \item Query $\Phi$ to obtain 
        $\p{U^{\otimes n}}^\dagger
            \Phi\p{
                U^{\otimes n}
                \ketbra{i}_S\otimes \ketbra{j}_{S^c}
                \p{U^{\otimes n}}^\dagger
            }
        U^{\otimes n}$;

    \item Measure qubits in $S$ over the computational basis, 
        set $Y=0$ if the result is $i$, otherwise set $Y=1$;

    \item Return $Y$.
\end{algorithmic}
\end{myalg}

\begin{restatable}{theorem}{thmrie}\label{thm:rie}
    Given quantum channel $\Phi$ and $S\subseteq [n]$, let $Y$ be the output of Algorithm $\ref{alg:rie}$. For arbitrary $\delta>0$, we have:
    \begin{enumerate}
    	\item $\inf_S[\Phi]=0\Rarr Y=0$ with probability $1$;
    	\item $\inf_S[\Phi]\ge \delta\Rarr \E{Y}
           \ge \f{2}{3}\cdot \inf_S[\Phi]
           \ge \f{2}{3}\cdot \delta$.
    \end{enumerate}
\end{restatable}

\begin{proof}
    In the first case, 
    when $\inf_S[\Phi]=0$, we know $\Phi=\tilde\Phi_{S^c}\otimes I_{S}$ 
    for some $\tilde\Phi_{S^c}$, therefore $Y=0$ with probability $1$. 
    We focus on the second case.
    \begin{align*}
        \Prob{Y=0\mid U=I} &= \f{1}{2^n}\sum_{i\in \{0,1\}^S}\sum_{j\in \{0,1\}^{n-S}} \Prob{Y=0\mid i,j,U=I} \\
        &= \f{1}{2^n}\sum_{i}\sum_{j} \Tr\ \p{\bra{i}_S\otimes I_{S^c}}\cdot \Phi(\ketbra{i,j})\cdot \p{\ket{i}_S\otimes I_{S^c}} \\
        &= \f{1}{2^n}\sum_i \sum_j \sum_{x,y\in \Zfn} \Tr \bra{i}\otimes I \cdot \wPhi(x,y)\sigma_x \ketbra{i,j} \sigma_y \cdot \ket{i}\otimes I \\
        &= \f{1}{2^n}\sum_i \sum_j \sum_{x,y\in \Zfn} \wPhi(x,y) \bra{i} \sigma_{x_S} \ket{i} \bra{i}\sigma_{y_S} \ket{i} \cdot \bra{j} \sigma_{y_{S^c}}\sigma_{x_{S^c}} \ket{j} \\
        &= \f{1}{2^n}\sum_{x,y\in \Zfn} \wPhi(x,y) \p{\sum_i  \bra{i} \sigma_{x_S} \ket{i} \bra{i}\sigma_{y_S} \ket{i}} \cdot \p{\sum_j \bra{j} \sigma_{y_{S^c}}\sigma_{x_{S^c}} \ket{j}}.
    \end{align*}

    For the summation in the first bracket, We note that if $x_S\notin \{0,3\}^S$, say, $x_S$ contains $1$ or $2$, $\bra{i} \sigma_{x_S} \ket{i}=0$ for all $i$. Same for $y_S$. Thus if $x_S\neq y_S\in \{0,3\}^S$, then $\sum_i  \bra{i} \sigma_{x_S} \ket{i} \bra{i}\sigma_{y_S} \ket{i}=0$.  For the summation in the second bracket, we have $\sum_j \bra{j} \sigma_{y_{S^c}}\sigma_{x_{S^c}} \ket{j}=\ip{\sigma_{y_{S^c}},\sigma_{x_{S^c}}}$, which is zero if $y_{S^c}\neq x_{S^c}$ and $2^{n-S}$ if $y_{S^c}=x_{S^c}$. Hence,
    \begin{align*}
    	\Prob{Y=0\mid U=I} &= \f{1}{2^n}\sum_{x,y\in \Zfn} \wPhi(x,y) \p{\sum_i  \bra{i} \sigma_{x_S} \ket{i} \bra{i}\sigma_{y_S} \ket{i}} \cdot \p{\sum_j \bra{j} \sigma_{y_{S^c}}\sigma_{x_{S^c}} \ket{j}}\\
    	&= \sum_{x\in \Zfn; x_{S}\in \{0,3\}^{S}} \wPhi(x,x).
    \end{align*}

    Note that $H^\dagger \sigma_0 H = \sigma_0$, $H^\dagger \sigma_1 H = \sigma_3$,
    $H^\dagger \sigma_2 H = -\sigma_2$, $H^\dagger \sigma_3 H = \sigma_1$. 
    We use $R_x$ as an alias for $R_x\p{\f{\pi}{2}}$ used in the algorithm above
    and have $R_x^\dagger \sigma_0 R_x = \sigma_0$, $R_x^\dagger \sigma_1 R_x = \sigma_1$,
    $R_x^\dagger \sigma_2 R_x = -\sigma_3$, $R_x^\dagger \sigma_3 R_x = \sigma_2$.
    With similiar calculations as above, we obtian
    \begin{align*}
        &\Prob{Y=0\mid U=H} = 
            \sum_{x\in \Zfn; x_{S}\in \{0,1\}^{S}} \wPhi(x,x), \\
        &\Prob{Y=0\mid U=R_x} = 
            \sum_{x\in \Zfn; x_{S}\in \{0,2\}^{S}} \wPhi(x,x).
    \end{align*}

    Recall that
    $$
        1 - \inf_{S}[\Phi] = \sum_{x\in \Zfn; x_{S}=\0} \wPhi(x,x).
    $$
    We can deduce that
    \begin{align*}
        \E{Y}&=\f{1}{3}\p{\Prob{Y=1\mid U=I} + \Prob{Y=1\mid U=H} + \Prob{Y=1\mid U=R_x}} \\
             &\ge \f{1}{3}\cdot 2 \cdot \inf_S[\Phi].
    \end{align*}
    The conclusion follows.

\end{proof}

\section{Testing Quantum $k$-Junta Channels}\label{app:testing}

\subsection{$O(k)$ Upper Bound on Testing $k$-Junta Quantum Channels, Proof of Theorem~\ref{thm:testup2}}

Firstly we show a key property of Algorithm~\ref{alg:is}, \algname{Influence-Sample}.

\begin{lemma}[Property of Algorithm~\ref{alg:is}]\label{lemma:keyproperty}
    For every iteration in $\algname{Influence-Sample}(\Phi,t)$, 
        we have with probability at least $\f{2}{3}\cdot\inf_{S^c}[\Phi]$, 
        $S$ become larger in this iteration.
\end{lemma}

\begin{proof}
    Following the similar calculation of the proof pf Theorem~\ref{thm:rie},
        recall that $i$ is the random string sampled in the iteration of 
        \algname{Influence-Sample} and $i'$ is the corresponding measurement 
        result, we have 
    \begin{align*}
        \Prob{\text{$S$ becomes larger}} &= 
            \Prob{\exists l, l\in [n]\backslash S, i_l\neq i'_l} \\
        &= \Prob{\algname{Influence-Estimator}(\Phi, S^c) = 1} \\
        &\ge \f{2}{3}\cdot \inf_{S^c}[\Phi].
    \end{align*}
\end{proof}

\thmtestuptwo*

\begin{proof}
    Recall that $S$ is the output of 
        $\algname{Influence-Sample}(\Phi,60(k+1)/\vep^2)$.
    If $\Phi$ is a $k$-junta channel, let $R$ be the qubits in which $\Phi$ 
        has non-trivial effect. We have $\abs{R}\le k$.
    It is easy to see $S\subseteq R$ with probability $1$, therefore 
        the algorithm always outputs ``Yes'' in this case.
    
    For the case $\Phi$ is $\vep$-far from any $k$-junta channel, 
        we let $X_i$, $i\in \{0,1,\dots,k\}$ be the number of iterations between
        when the $i$-th coordinate is added to $S$ and when the $i+1$-th 
        coordinate is added to $S$. 
    As an example of $n=5$, $S$'s are $\varnothing,\{1,2\},\{1,2,4\}$ 
        through three iterations, we have $X_0=2$, $X_1=0$, $X_2=1$.
    According to Corollary~\ref{corr:keythm}, for any $\abs{S}\le k$, 
        $\inf_{S^c}[\Phi]\ge \vep^2/4$.
    Combining with Lemma~\ref{lemma:keyproperty}, which says when $\abs{S}=i$,
        for each iteration, $S$ becomes largers with probability at least 
        $\f{2}{3}\cdot\inf_{S^c}[\Phi]\ge \f{\vep^2}{6}$. Therefore, 
        $\E{X_i}\le \f{6}{\vep^2}$, for $\forall i\in \{0,1,\dots,k\}$.
    We conclude that
    $$
        \E{\sum_{i=0}^kX_i}\le \f{6(k+1)}{\vep^2}.
    $$
    By Markov Inequality, 
    $$
        \Prob{\sum_{i=0}^kX_i\ge \f{60(k+1)}{\vep^2}} \le \f{1}{10}.
    $$
    Thus in the case $\Phi$ is $\vep$-far from any $k$-junta channel, 
        the algorithm outputs ``No'' with probability at least $9/10$.

    Besides, $\algname{Junta-Channel-Tester}(\Phi,k,\vep)$ makes 
        $O(k/\vep^2)$ queries to $\Phi$. 
    Theorem~\ref{thm:testup2} follows.

\end{proof}

\subsection{$\Omega(\sqrt{k})$ Lower Bound on Testing Quantum $k$-Junta Channels, Proof of Lemma~\ref{lemma:lokey} and Lemma~\ref{lemma:lored}}

Before proving Lemma~\ref{lemma:lokey}, we need the following technical lemma.
\begin{lemma}\label{L:border claim}
    Let $n,m$ be natural numbers and $r_{a,b}\in \RR, r_{a,a}\ge 0$ for $a,b\in[n]$. The maximum value
    \begin{align*}
        \max_x\ &\sum_{a,b\in[n]} r_{a,b}x_ax_b \\
        s.t.\ &x_a\in[-m,m] , \forall a\in[n]
    \end{align*}
    can be achieved if we restrict $x$ satisfying that $x_a\in \{m,-m\}$, for all $a\in [n]$.
\end{lemma}
\begin{proof}
    For arbitrary $a\in [n]$,
    $$
        \f{\partial^2\sum_{a,b\in[n]}r_{a,b}x_ax_b}{\p{\partial x_a}^2} = 2r_{a,a}\ge 0.
    $$
    Thus the objective function is convex in $x_a$ for all $a\in[n]$. The conclusion follows.

\end{proof}

We will prove our key technical lemma, Lemma~\ref{lemma:lokey}.

\lemmalokey*

\begin{proof}
    For any $k$-junta channel $\Phi$, let boolean function $g$ minimize $D(\Phi,\Phi_g)$. We will show $g$ could be a $k$-junta boolean function.
    \begin{align*}
        g &= \argmin	_{g} D(\Phi,\Phi_g) \\
        &= \argmin_{g} \norm{\sum_{a,b\in \{0,1\}^n} \p{\Phi(\ketbra{a}{b})-\Phi_g(\ketbra{a}{b})}\otimes \ketbra{a}{b}}_2^2 \\
        &= \argmin_{g} \sum_{a,b\in \{0,1\}^n} \norm{ \Phi(\ketbra{a}{b})-\Phi_g(\ketbra{a}{b})}_2^2 \\
        &= \argmin_{g} \sum_{a,b\in \{0,1\}^n} \p{ \bra{a}\Phi(\ketbra{a}{b})\ket{b}-(-1)^{g(a)+g(b)}}^2
    \end{align*}
    Since $\Phi$ is a $k$-junta channel, there exists $T\subseteq [n]$,$|T|=k$, such that $\Phi(\rho)=\widetilde{\Phi}(\rho_T)\otimes \rho_{T^c}$. We have
    \begin{align*}
        g &= \argmin_{g} \sum_{a,b\in \{0,1\}^n} \p{ \bra{a}\Phi(\ketbra{a}{b})\ket{b}-(-1)^{g(a)+g(b)}}^2 \\
        &= \argmin_{g} \sum_{a',b'\in \{0,1\}^T\atop a'',b''\in \{0,1\}^{T^c}} \p{ \bra{a',a''}\widetilde{\Phi}(\ketbra{a'}{b'})\otimes \ketbra{a''}{b''}\ket{b',b''}-(-1)^{g(a',a'')+g(b',b'')}}^2 \\
        &= \argmin_{g} \sum_{a',b'\in \{0,1\}^T\atop a'',b''\in \{0,1\}^{T^c}} \p{ \bra{a'}\widetilde{\Phi}(\ketbra{a'}{b'})\ket{b'}-(-1)^{g(a',a'')+g(b',b'')}}^2 \\
        &= \argmax_{g}\sum_{a',b'\in \{0,1\}^T\atop a'',b''\in \{0,1\}^{T^c}} \p{ \bra{a'}\widetilde{\Phi}(\ketbra{a'}{b'})\ket{b'}+\overline{\bra{a'}\widetilde{\Phi}(\ketbra{a'}{b'})\ket{b'}}}\cdot (-1)^{g(a',a'')}\cdot (-1)^{g(b',b'')} \\
        &= \argmax_{g} \sum_{a',b'\in \{0,1\}^T} \p{ \bra{a'}\widetilde{\Phi}(\ketbra{a'}{b'})\ket{b'}+\overline{\bra{a'}\widetilde{\Phi}(\ketbra{a'}{b'})\ket{b'}}}\cdot g'(a')\cdot g'(b')
    \end{align*}
    where $g'(a')=\sum_{a''\in\{0,1\}^{T^c}}(-1)^{g(a',a'')}$,$a'\in\{0,1\}^T$. Let $r_{a',b'}=2\text{Re}\p{\bra{a'}\widetilde{\Phi}(\ketbra{a'}{b'})\ket{b'}}$, we know $r_{a',a'}\ge 0$ since $\widetilde{\Phi}(\ketbra{a'}{b'})$ is PSD.

    Combining with Lemma~\ref{L:border claim}, we know that there exists $g'$ which achieves the maximum satisfying that $g'(a')=2^{n-k}$ or $g'(a')=-2^{n-k}$, for all $a'\in \{0,1\}^T$. Thus, we can take $k$-junta boolean function $g$ to obtain the minimum of $D(\Phi,\Phi_g)$.
\end{proof}

Before we prove Lemma~\ref{lemma:lored}, we need the following lemma.

\begin{lemma}\label{L:lemma1}
    Given boolean function $f$, if $f$ is $\vep$-far from any $k$-junta boolean function, then for any $k$-junta boolean function $g$, we have $D(\Phi_f,\Phi_g)\ge \sqrt{2\vep}$.
\end{lemma}

\begin{proof}
    For $k$-junta boolean function $g$, if $1/2\ge D(f,g)\ge \vep$, we claim that $D(\Phi_f,\Phi_g)\ge \sqrt{2D(f,g)}\ge \sqrt{2\vep}$. Recall that $D(f,g)=\Pr_{x}[f(x)\neq g(x)]$.
    \begin{align*}
        D(\Phi_f,\Phi_g)&=\f{1}{\sqrt{2}N}\norm{\vect{U_f}\vect{U_f}^*-\vect{U_g}\vect{U_g}^*}_2	 \\
        &=\f{1}{\sqrt{2}N}\sqrt{\sum_{a,b\in\{0,1\}^n}\lrb{(-1)^{f(a)+f(b)}-(-1)^{g(a)+g(b)}}^2} \\
        &=\f{\sqrt{2}}{N}\sqrt{\sum_{a,b\in\{0,1\}^n}\p{\mathbbm{1}[f(a)\neq g(a)]\cdot \mathbbm{1}[f(b) = g(b)]+\mathbbm{1}[f(a)= g(a)]\cdot \mathbbm{1}[f(b) \neq g(b)]}} \\
        &=\f{\sqrt{2}}{N}\sqrt{2(1-D(f,g))D(f,g)\cdot N^2} \\
        &\ge \sqrt{2\vep}
    \end{align*}
\end{proof}

\begin{figure}
  \centering
  \input{tikz/lowerbound.tikz}
  \caption{Illustration of proof of Lemma~\ref{lemma:lored}}
\label{fig:lowerbound}
\end{figure}
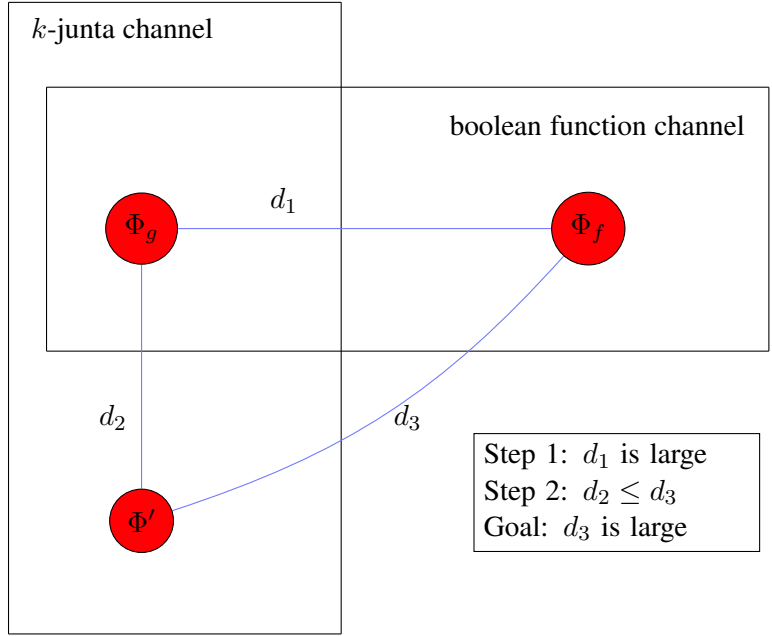

\lemmalored*

\begin{proof}
	We will analyze the output of a $(k,\sqrt{\vep/2})$-channel junta tester given oracle to $\Phi_f$ for some boolean function $f:\{0,1\}^n\rarr \{0,1\}$.
	
	If $f$ is a $k$-junta, it is easy to see $\Phi_f$ is also a $k$-junta. Thus the channel junta tester outputs ``Yes'' with probability at least $9/10$.
	
	If $f$ is $\vep$-far from any $k$-junta boolean function, we are going to show $\Phi_f$ is $\sqrt{\vep/2}$-far from any $k$-junta channel. We give an illustration of our proof as Figure~\ref{fig:lowerbound}. Our goal is show for any $k$-junta channel $\Phi'$, $D(\Phi_f,\Phi')=d_3$ is large. We firstly show, from Lemma~\ref{L:lemma1}, that for any $k$-junta boolean function $g$, $D(\Phi_f,\Phi_g)=d_1$ is also large (as Step 1 in figure~\ref{fig:lowerbound}). Next we show for any $k$-junta channel $\Phi'$, there exists $k$-junta boolean function $g$ such that $D(\Phi',\Phi_g)\le D(\Phi_f,\Phi')$, i.e., $d_2\le d_3$. Finally, we conclude that $d_3\ge (d_2+d_3)/2\ge d_1$.

    For any $k$-junta channel $\Phi'$, let $g=\argmin_{g\text{ is a boolean function}}D(\Phi',\Phi_g)$. Then $g$ is a $k$-junta by Lemma~\ref{lemma:lokey}. We have
    $$
        D(\Phi',\Phi_g)\le D(\Phi',\Phi_f)
    $$
    and since $g$ is a $k$-junta, according to Lemma~\ref{L:lemma1},
    $$
        D(\Phi_f,\Phi_g) \ge \sqrt{2\vep}
    $$
    To conclude, we have
    $$
        D(\Phi',\Phi_f)\ge \f{1}{2} \p{D(\Phi',\Phi_f)+D(\Phi',\Phi_g)} \ge \f{1}{2} D(\Phi_f,\Phi_g) \ge \sqrt{\f{\vep}{2}}
    $$
    for any $k$-junta channel $\Phi'$.

\end{proof}

\subsection{Reduction from $k$-Junta Unitary Testing, Proof of Lemma~\ref{lemma:redkey} and Lemma~\ref{lemma:red2}}

The proof of Lemma~\ref{lemma:redkey} follows the same line as Lemma~\ref{lemma:lokey}. We first show a lemma similar to Lemma~\ref{L:border claim}, which will be used later.

\begin{lemma}\label{RtUJT:border claim}
    Let $n$ be a natural number. For $a,b\in[n]$, let $A_{a,b}\in \CC^{n\times n}$ be an $n\times n$ matrix. Set $A=\sum_{a,b\in [n]}A_{a,b}\otimes \ketbra{a}{b}\in \CC^{n^2\times n^2}$ to be the Choi representation of a quantum channel $\Phi$. In other words, $A$ is PSD and there exists $(B_s)_s$, $B_s\in \CC^{n\times n}$, s.t., $A=\sum_s \vect{B_s}\vect{B_s}^*$ and $\sum_s B_s^*B_s=I$. The maximum value
    \begin{align*}
        \max_{V\in\CC^{n\times n}}\ &\sum_{a,b\in[n]} \bra{a}V^*A_{a,b}V\ket{b} \\
        s.t.\ & V^*V \le I
    \end{align*}
    can be achieved if we restrict $V$ satisfying that $V^*V=I$.
\end{lemma}

\begin{proof}
    Note that
    \begin{align*}
        	\sum_{a,b\in[n]} \bra{a}V^*A_{a,b}V\ket{b} &= \vect{V}^*A\vect{V}\\
        	&= \sum_s \vect{V}^*\vect{B_s}\vect{B_s}^*\vect{V} \\
        	&= \sum_s \abs{\ip{V,B_s}}^2,
    \end{align*}
    where $V$ takes over all matrices in
    \begin{eqnarray*}
      &&\{V\mid V\in \CC^{n \times n},V^*V\le I\}\\
      &&\hspace{1.8cm}=\{W\Sigma W'\mid W,W'\text{ are unitaries}, \Sigma\text{ is a diagonal real matrix},-I\le\Sigma\le I\} \end{eqnarray*}
    by the SVD decomposition. Suppose $\Sigma=\textsf{Diag}(x_1,\ldots,x_n)$. It is not hard to see
    \[\sum_s \abs{\ip{V,B_s}}^2=\sum_s \abs{\ip{W\Sigma W'\mid W,B_s}}^2\]
    is a quadratic form in $x_1,\ldots,x_n$. The coefficients of $x_a$ is $B_{s,aa}'^2\ge 0$, where $B_s'=W^*B_sW'^*$. By Lemma~\ref{lemma:redkey}, the maximum can be achieved if $x_a=\pm 1$. We conclude the result.
%
%
%
%
%
\end{proof}

\lemmaredkey*

\begin{proof}
    For any $k$-junta channel $\Phi'$, let unitary $V$ minimize $D(\Phi',\Phi_V)$. We will show $V$ could be a $k$-junta unitary.
    \begin{align*}
        V &= \argmin	_{V} D(\Phi',\Phi_V) \\
        &= \argmin_{V} \norm{\sum_{a,b\in \{0,1\}^n} \p{\Phi'(\ketbra{a}{b})-\Phi_V(\ketbra{a}{b})}\otimes \ketbra{a}{b}}_2^2 \\
        &= \argmin_{V} \sum_{a,b\in \{0,1\}^n} \norm{ \Phi'(\ketbra{a}{b})-\Phi_V(\ketbra{a}{b})}_2^2 \\
        &= \argmin_{V} \sum_{a,b\in \{0,1\}^n} \norm{ \Phi'(\ketbra{a}{b})-V\ketbra{a}{b}V^*}_2^2
    \end{align*}
    Since $\Phi'$ is a $k$-junta, there exists $T\subseteq[n]$,$|T|=k$, s.t., $\Phi'=\widetilde\Phi'_T\otimes I_{T^c}$. Let $V = \sum_{x\in \ZZ_4^{T^c}}V_x\otimes \sigma_x$. Besides, because $\sum_{a,b}\norm{ V\ketbra{a}{b}V^*}_2^2=4^n$, we have:
    \begin{align*}
        V &= \argmin_{V} \sum_{a,b\in \{0,1\}^n} \norm{ \Phi'(\ketbra{a}{b})-V\ketbra{a}{b}V^*}_2^2 \\
        &= \argmax_{V}\sum_{a,b\in\{0,1\}^n} \text{Re }\lrs{\bra{b}V^*\Phi'(\ketbra{b}{a})V\ket{a}} \\
        &= \argmax_{V}\sum_{a,b\in\{0,1\}^n} \bra{b}V^*\Phi'(\ketbra{b}{a})V\ket{a} \\
        &= \argmax_{V}\sum_{a',b'\in\{0,1\}^T}\sum_{a'',b''\in\{0,1\}^{T^c}}\sum_{x,y\in \ZZ_4^{T^c}} \bra{b'}V_x^*\widetilde\Phi'(\ketbra{b'}{a'})V_y\ket{a'} \cdot \bra{b''}\sigma_x\ket{b''}\cdot \bra{a''}\sigma_y\ket{a''} \\
        &= \argmax_{V}\sum_{a',b'\in\{0,1\}^T}\sum_{x,y\in \ZZ_4^{T^c}} \bra{b'}V_x^*\widetilde\Phi'(\ketbra{b'}{a'})V_y\ket{a'} \cdot \Tr \sigma_x\cdot \Tr\sigma_y \\
        &= \argmax_{V}\sum_{a',b'\in\{0,1\}^T}\bra{b'}V_{0^{T^c}}^*\widetilde\Phi'(\ketbra{b'}{a'})V_{0^{T^c}}\ket{a'}
    \end{align*}
    where the second equality follows from
    \begin{align*}
        \norm{ \Phi'(\ketbra{a}{b})-V\ketbra{a}{b}V^*}_2^2 &= \norm{\Phi'(\ketbra{a}{b})}_2^2 + \norm{V\ketbra{a}{b}V^*}_2^2 - 2\cdot \text{Re }\lrs{\bra{b}V^*\Phi'(\ketbra{b}{a})V\ket{a}} \\
        &= \norm{\Phi'(\ketbra{a}{b})}_2^2 + \norm{\ketbra{a}{b}}_2^2 - 2\cdot \text{Re }\lrs{\bra{b}V^*\Phi'(\ketbra{b}{a})V\ket{a}}
    \end{align*}
    and the first two terms have nothing to do with $V$. The third equality is because, if $a=b$,  $\bra{b}V^*\Phi'(\ketbra{b}{a})V\ket{a}\in \RR$ and $\text{Re }\lrs{\bra{b}V^*\Phi'(\ketbra{b}{a})V\ket{a}} = \bra{b}V^*\Phi'(\ketbra{b}{a})V\ket{a}$. If $a\neq b$,
    \begin{align*}
        \sum_{a,b\in\{0,1\}^n, a\neq b} \bra{b}V^*\Phi'(\ketbra{b}{a})V\ket{a} &= \sum_{a,b\in\{0,1\}^n, a< b} \bra{b}V^*\Phi'(\ketbra{b}{a})V\ket{a} + \bra{a}V^*\Phi'(\ketbra{a}{b})V\ket{b} \\
        &= \sum_{a,b\in\{0,1\}^n, a< b} \bra{b}V^*\Phi'(\ketbra{b}{a})V\ket{a} + \overline{\bra{b}V^*\Phi'(\ketbra{b}{a})V\ket{a}} \\
        &= \sum_{a,b\in\{0,1\}^n, a< b} 2\cdot \text{Re}\lrs{\bra{b}V^*\Phi'(\ketbra{b}{a})V\ket{a}} \in \RR
    \end{align*}
    Notice that $\Tr_{T^c}V^*V=\sum_{x,y\in\ZZ_4^{T^c}}V_x^*V_y\ip{\sigma_x,\sigma_y}=2^{n-k}\sum_{x\in\ZZ_4^{T^c}}V_x^*V_x=2^{n-k}I_{T}$. By Lemma~\ref{RtUJT:border claim}, the maximum can be achieved when $V_{0^{T^c}}^*V_{0^{T^c}}=I_T$, which implies $V_x=0$ for $x\neq 0^{T^c}$. Thus, we can take $k$-junta unitary $V$ to obtain the minimum of $D(\Phi',\Phi_V)$.
\end{proof}

The proof of Lemma~\ref{lemma:red2} is similar to Lemma~\ref{lemma:lored}.

\lemmaredtwo*

\begin{proof}
	Let $U$ be a unitary matrix. It suffices to show that if $U$ is $\vep$-far from any $k$-junta unitary, then $\Phi_U$ is $\vep/2$-far from any $k$-junta channel. We firstly show that for any $k$-junta unitary $V$, $D(\Phi_U,\Phi_V)$ is large. Next, we show for any $k$-junta channel $\Phi'$, there exists $k$-junta unitary $V$ such that $D(\Phi',\Phi_V)\le D(\Phi',\Phi_U)$.

	For any $k$-junta channel $\Phi'$, let $V=\argmin_{\text{unitary }V}D(\Phi',\Phi_V)$. By Lemma~\ref{lemma:redkey}, $V$ is $k$-junta. For any unitary $U$, if $U$ is $\vep$-far from any $k$-junta unitary, then  therefore $D(\Phi_U,\Phi_V)\ge \vep$. Thus
	$$
	    D(\Phi_U,\Phi')\ge \f{1}{2}(D(\Phi_U,\Phi')+D(\Phi',\Phi_V)) \ge \f{1}{2}D(\Phi_U,\Phi_V)\ge \vep/2.
	$$
	The last inequality follows directly from Lemma~\ref{lemma:reldis}. We complete the proof.
	
\end{proof}

\section{$O(4^k/\vep^2)$ Upper Bound on Learning Quantum $k$-Junta Channels, Proof of Theorem~\ref{thm:plearner}}\label{app:learning}

Before describing the learning algorithm \algname{Junta-Channel-Learner}, we introduce a tomography algorithm from~\cite{OW17_tomo}.

\begin{fact}[Corollary 1.4 of~\cite{OW17_tomo}]
	There exists an algorithm \algname{Tomography}, which is given $O(4^k/\vep^2)$ copies of an unknown $2k$ qubit state $\rho$ and outputs the description an estimated state $\tilde{\rho}$ satisfying that $\|\rho-\tilde{\rho}\|_2\leq\vep$, with probability at least $0.99$.
\end{fact}

Now we are ready to prove Theorem~\ref{thm:plearner}.

\thmplearner*

\begin{proof}
	Let $R$ be the subset of $[n]$, over which $\Phi$ acts non-trivially. 
    Recall that $S$ is the output of the call to Algorithm~\ref{alg:is}, 
        $\algname{Influence-Sample} $ in line 1 of 
        $\algname{Junta-Channel-Learner}(\Phi,k,\vep)$.
    It is easy to see $S\subseteq R$ with probability $1$. 
    From a similar calculation as the proof of Lemma~\ref{lemma:keyproperty},
        for each $i\in R$ with $\inf_i[\Phi]\ge \f{\vep^2}{8k}$,
        the probability that $i$ is not added to $S$ in one interation is 
        at most $\f{2}{3}\cdot \f{\vep^2}{8k}$, therefore
    $$
        \Prob{i\notin S}\le \p{1-\f{2}{3}\cdot\f{\vep^2}{8k}}^{O(k\log k/\vep^2)}
        \le \f{1}{100k}.
    $$
    Thus with probability at least $0.99$,
        for all $i\in R$ with $\inf_i[\Phi]\ge \f{\vep^2}{8k}$,
        we have $i\in S$,
    and with probability at least $0.99$, 
        $\inf_{R-S}[\Phi]\le \f{\vep^2}{8k}\cdot \abs{R}\le \f{\vep^2}{8}$.
	
	Let $\Phi=\Phi_R\otimes I^{R^c}$ and $v(\Phi)=v(\Phi_R)\otimes v(I^{R^c})$. Consider the quantum state $\psi$ returned by \algname{Quantum-State-Preparation}. We conclude that the probability that it does not output ``error'' is
	\begin{align*}
	    \Tr v(\Phi)\cdot \left(I^S\otimes \ketbra{v(I^{S^c})}\right) &= \sum_{x,y\in \Zfn} \wPhi(x,y) \Tr v(\Phi_{x_S,y_S})\cdot  \bra{v(\sigma_\0)} v(\Phi_{x_{S^c},y_{S^c}}) \ket{v(\sigma_\0)}  \\
	    &= \sum_{x\in \Zfn,x_{S^c}=\0}\wPhi(x,x) \\
	    &= \inf_S[\Phi] \ge 1-\vep^2/8.
	\end{align*}
	The step 2 of \algname{Junta-Channel-Learner} collects $t$ copies of $\psi$ in $10t$ calls to the preparation subroutine with probability at least $0.99$ for large enough $k$, since the expectation of successful collections is at least $(1-\vep^2/8)\cdot 10t\ge 8t$. We will show $\psi\otimes v(I^{R-S})$ is close to $v(\Phi_R)$.
	
    It is easy to calculate that
    \begin{align*}
        \psi &= \f{1}{\inf_S{\Phi}}I^S\otimes \bra{v(\sigma_\0)} \cdot v(\Phi) \cdot I^S\otimes \ket{v(\sigma_\0)} \\
        &= \f{1}{\inf_S{\Phi}} \sum_{x,y\in \Zfn} \wPhi(x,y)v(\Phi_{x_S,y_S})\cdot \bra{v(\sigma_\0)} \cdot v(\Phi_{x_{S^c},y_{S^c}}) \cdot \ket{v(\sigma_\0)} \\
        &=  \f{1}{\inf_S{\Phi}} \sum_{x,y\in \Zfn,x_{S^c}=y_{S^c}=\0} \wPhi(x,y)v(\Phi_{x_S,y_S})
    \end{align*}

    Let $\psi'=\inf_S[\Phi]\cdot \psi = \sum_{x,y\in \Zfn,x_{S^c}=y_{S^c}=\0} \wPhi(x,y)v(\Phi_{x_S,y_S})$. We have
    \begin{align*}
        \norm{\psi \otimes v(I^{S^c}) - v(\Phi)}_2 &\le  \norm{\psi \otimes v(I^{S^c}) - \psi'\otimes v(I^{S^c})}_2 + \norm{\psi'\otimes v(I^{S^c})-v(\Phi)}_2 \\
        &= \p{1 - \inf_S[\Phi]}\norm{\psi\otimes v(I^{S^c})}_2 + \norm{\psi'\otimes v(I^{S^c})-v(\Phi)}_2 \\
        &\le \f{\vep^2}{8} + \norm{\psi'\otimes v(I^{S^c})-v(\Phi)}_2
    \end{align*}
    and
    \begin{align*}
    	\norm{\psi'\otimes v(I^{S^c})-v(\Phi)}_2^2 &= \norm{\sum_{x,y\in \Zfn,x_{S^c}\neq \0\text{ or }y_{S^c}\neq \0} \wPhi(x,y)v(\Phi_{x,y})}_2^2 \\
    	&= \sum_{x,y\in \Zfn,x_{S^c}\neq \0\text{ or }y_{S^c}\neq \0} \abs{\wPhi(x,y)}^2 \\
    	&\le \sum_{x,y\in \Zfn,x_{S^c}\neq \0\text{ or }y_{S^c}\neq \0} \wPhi(x,x)\wPhi(y,y) \\
    	&\le 2\sum_{x,y\in \Zfn,x_{S^c}\neq \0} \wPhi(x,x)\wPhi(y,y) \\
    	&= 2\inf_{S^c}[\Phi] \le \f{\vep^2}{4}
    \end{align*}

    Therefore
    $$
        \f{1}{\sqrt{2}}\norm{\psi \otimes v(I^{S^c}) - v(\Phi)}_2\le \f{1}{\sqrt{2}}\cdot \p{\f{\vep^2}{8} + \f{\vep}{2}} \le 0.45\vep
    $$

	By the step 1 of \algname{Tomography}, we get a description of quantum state $\phi$ with probability  $0.99$ s.t. $\norm{\phi-\psi}_2\le 0.04\vep$ and $\norm{\phi\otimes v(I^{S^c})- v(\Phi)}_2/\sqrt{2}\le 0.49\vep$. After we find the closest Choi state $\phi'$ to $\phi$ in the step 2 of \algname{Tomography}, we are sure that $\norm{\phi'\otimes v(I^{S^c})- v(\Phi)}_2\le \vep$ and the returned channel is close to $\Phi$ with  distance at most $\vep$ with probability at least $9/10$.
	
	To see the query complexity, the call to \algname{Pauli-Sample} costs only $O\p{k\log k/\vep^2}$ queries to $\Phi$ and the preparation and tomography need $O\p{4^k/\vep^2}$ queries. The total queries are $O\p{4^k/\vep^2}$.
	
\end{proof}

\end{document}

%% file: tikz/sample_style.tikzstyles
\tikzstyle{red}=[fill={rgb,255: red,255; green,0; blue,4}, draw=black, shape=circle]

\tikzstyle{blue edge}=[-, fill=none, draw={rgb,255: red,111; green,118; blue,255}]

%% file: tikz/lowerbound.tikz
\begin{tikzpicture}
	\begin{pgfonlayer}{nodelayer}
		\node [style=none] (1) at (-6.25, 5) {};
		\node [style=none] (2) at (2.5, 5) {};
		\node [style=none] (3) at (-6.25, -11.75) {};
		\node [style=none] (4) at (2.5, -11.75) {};
		\node [style=none] (5) at (-5.25, 2.75) {};
		\node [style=none] (6) at (13.75, 2.75) {};
		\node [style=none] (7) at (-5.25, -4.25) {};
		\node [style=none] (8) at (13.75, -4.25) {};
		\node [style=red] (9) at (-2.75, -1) {$\Phi_g$};
		\node [style=red] (10) at (9, -1) {$\Phi_f$};
		\node [style=red] (11) at (-2.75, -8.75) {$\Phi'$};
		\node [style=none] (13) at (1, -0.25) {$d_1$};
		\node [style=none] (14) at (4.25, -6) {$d_3$};
		\node [style=none] (15) at (-3.5, -6) {$d_2$};
		\node [style=none] (17) at (9.25, 1.75) {boolean function channel};
		\node [style=none] (18) at (-3.25, 4.25) {$k$-junta channel};
		\node [draw, text width=3.5cm] (19) at (9.75, -8) {Step 1: $d_1$ is large \\ Step 2: $d_2\le d_3$\\ Goal: $d_3$ is large};
	\end{pgfonlayer}
	\begin{pgfonlayer}{edgelayer}
		\draw (1.center) to (2.center);
		\draw (1.center) to (3.center);
		\draw (3.center) to (4.center);
		\draw (4.center) to (2.center);
		\draw (5.center) to (6.center);
		\draw (5.center) to (7.center);
		\draw (7.center) to (8.center);
		\draw (8.center) to (6.center);
		\draw [style=blue edge] (9) to (10);
		\draw [style=blue edge, bend left=15] (10) to (11);
		\draw [style=blue edge] (11) to (9);
	\end{pgfonlayer}
\end{tikzpicture}